\documentclass[12pt]{article}

\usepackage[utf8]{inputenc}
\usepackage[T1]{fontenc}
\usepackage[letterpaper, margin=3.5cm]{geometry}
\usepackage{url}
\usepackage{graphicx}
\usepackage{subcaption}
\usepackage{wrapfig}
\usepackage{amssymb, amsmath, amsthm, amsfonts}
\usepackage{mathtools}
\usepackage{tcolorbox}
\usepackage{multirow}
\usepackage{threeparttable}
\usepackage{setspace}
\usepackage{indentfirst}
\usepackage{hyperref}
\usepackage{theoremref}
\usepackage{soul}

\theoremstyle{definition}
\newtheorem{theorem}{Theorem}[section]
\newtheorem{proposition}{Proposition}[section]

\newtheorem{lemma}[theorem]{Lemma}

\AtBeginEnvironment{definition}{%
  \pushQED{\qed}%
}
\AtEndEnvironment{definition}{\popQED\endexample}
\newtheorem{notation}{Notation}[section]
\AtBeginEnvironment{notation}{%
  \pushQED{\qed}%
}
\newtheorem{example}{Example}[section]
\AtBeginEnvironment{example}{%
  \pushQED{\qed}%
}
\AtEndEnvironment{example}{\popQED\endexample}
\newtheorem{remark}{Remark}[section]
\newtheorem{assumption}{Assumption}[section]
\numberwithin{equation}{subsection}

\title{\textbf{Impact of Resistance Development Mechanisms on Antibiotic Treatment Outcomes}}
\author{Ailin Zhang, Shigui Ruan, and Xi Huo\footnote{Correspondence: x.huo@math.miami.edu}\\[2mm]
{\small Department of Mathematics, University of Miami, Coral Gables, FL 33146, USA}}
\date{}  % leave blank for no date

\begin{document}

\maketitle

\begin{abstract}
Bacteria develop resistance to antibiotics through various mechanisms, with the specific mechanism depending on the drug-bacteria pair. It remains unclear, however, which resistance mechanism best supports favorable treatment outcomes, specifically in clearing infections and inhibiting further resistance. In this study, we use periodic ordinary differential equation models to simulate different antibiotic treatment protocols for bacterial infections. Using stability analysis and numerical simulations, we investigate how different resistance mechanisms, including plasmid-induced and mutation-induced resistance, affect treatment outcomes. Our findings suggest that antibiotic treatments with fixed dosing schedules are more likely to be effective when resistance arises exclusively through plasmid-mediated transmission. Further, when treatment fails, mutation-driven mechanisms tend to favor the selection of fully resistant bacterial strains. We also investigated the efficacy of different treatment strategies based on these mechanisms, finding that a twice-daily regimen consistently outperforms a once-daily regimen in terms of infection clearance. Additionally, our simulations with short half-life antibiotics indicate that the "catch-up" strategy outperforms the "compensatory double-dose" approach after a missed dose, a finding that aligns with general pharmaceutical advice for short-half-life drugs.
\end{abstract}

\section{Introduction}
Over the past several decades, antibiotic-resistant bacteria (ARB) have emerged rapidly and spread worldwide. Infections caused by ARB are challenging to treat, leading to prolonged illness, disability, increased mortality, and higher healthcare costs. Each bacterial isolate causing infections may harbor a mix of antibiotic-resistant and susceptible populations, and the antibiotic concentration required to eliminate the bacterial population varies significantly based on the genetic material present. Resistance can further develop through various mechanisms during treatments. Therefore, the well-recognized "hit early and hit hard" concept emphasizes initiating treatment with an appropriate antibiotic dose and frequency so as to effectively clear the infection and prevent the development of resistance.
%ARB are present in humans, animals, plants, and the environment, moving freely between ecosystems. 

Bacteria develop resistance to antibiotics mostly by making genetic changes via two major mechanisms \cite{Hall2016,Hall2017}: chromosomal mutation - acquired by spontaneous mutations on the DNA, and horizontal gene transfer - acquired by bacteria-to-bacteria transmission of plasmids (small circular DNA). The specific mechanisms through which bacteria develop a certain resistance to an antibiotic depend highly on the bacteria-antibiotic combination. In modeling, chromosomal mutations are captured as a linear rate or a rate proportional to the logistic net growth rate 
\cite{Ankomah2012,Ankomah2014,Dasbasi2016,Handel2009,Levin2010,Lukacisinova2017,Pena-Miller2012}, mutations via plasmid transmission are characterized by mass action \cite{DAgata2008,DeLeenheer2010,Imran2007,Lopatkin2017,Ternent2015,Zhong2012}, and less work incorporate both mechanisms \cite{Hall2016,DDR2021}. For certain antibiotics, bacteria must obtain multiple genes to achieve full resistance. Thus some studies stratify the bacteria population into more than two compartments with different genotypic levels corresponding to various antibiotic killing rates \cite{Ankomah2014,DAgata2008,Zhong2012}. 
% which has been emphasized in most pharmacodynamics (PD) studies. In order to fit the \emph{in vitro} experimental data on the time-varying bacteria population count, PD models take extra consideration in tracking the concentration change of antibiotics during experiments, with the aim of learning the concentration-dependent killing rate of antibiotics on specific bacteria strains. For the sake of data fitting, such models simply categorize bacteria population as sensitive or resistant and usually model the mutation dynamics as a proportionality of the logistic net growth rate \cite{Campion2005,Chung2006,Firsov2003,Jumbe2003,Lee2017,Nguyen2014,Olofsson2005,Regoes2004,Tam2005,Wiuff2007,Zhi1988}.

Many studies have explored the impacts of treatment protocols on infection clearance efficacy and resistance development potential, focusing on specific resistance development mechanisms. However, there is a lack of research examining these impacts in the reverse direction, which could provide valuable insights for the theoretical guideline for drug-bug specific antibiotic treatments. For example, vancomycin and ciprofloxacin are both used to treat Methicillin-resistant {\it Staphylococcus aureus} (MRSA) infections, where vancomycin has been discovered decades earlier than ciprofloxacin. While it might be reasonable to assume that MRSA is more resistant to the older antibiotic due to longer exposure, MRSA is, in fact, more resistant to ciprofloxacin than vancomycin. Knowing that MRSA develops resistance to these drugs through different mechanisms — chromosomal mutations for ciprofloxacin and horizontal gene transfer for vancomycin — it is natural to ask how these mechanisms contribute to the observed discrepancy in resistance rates.

In this study, we develop and analyze two periodic ODE models of antibiotic treatment, distinguished by how bacteria acquire resistance: either through plasmid transmission or spontaneous chromosomal mutation. We examine the long-term dynamics of each model by deriving the basic reproduction number, $\mathcal{R}_0$, which characterizes infection clearance, and conducting uniform persistence analysis to identify conditions of treatment failure. Numerical simulations parameterized for MRSA infections treated with moxifloxacin, a fourth-generation fluoroquinolone, reveal three key findings: (1) conventional concepts such as the mutation selection window (MSW) and minimum inhibitory concentration (MIC) are insufficient to reliably predict treatment outcomes across resistance mechanisms; (2) treatment is more likely to succeed when resistance arises via plasmid transmission; and (3) under failed treatments, spontaneous mutation-driven resistance can lead to the selection of fully resistant strains.

The remainder of this paper is organized as follows. Section 2 introduces the baseline model without genetic alteration, followed by models in which bacteria acquire resistance through plasmid transmission and spontaneous mutation. For the latter two cases, we establish the existence and stability of periodic solutions. In Section 3, we parameterize the time-dependent bacterial killing rate and antibiotic concentration functions, as well as the concentration-dependent killing rate function. Section 4 presents numerical simulations, and Section 5 provides concluding discussions.
    
\section{Models}
\subsection{Baseline model with no genetic alteration}
Without the consideration of genetic alteration, the bacterial population can be modeled as one single compartment adopting a logistic growth model. Denote $B(t)$ as the bacterial population at time $t$, we have

\begin{equation}\tag{M1}\label{eq:M1}
\dot B(t) = r B(t)\left(1-\frac{B(t)}{K}\right) - \mu(t)B(t),
\end{equation}

where $r$ is the intrinsic bacterial growth rate. $K$ represents the carrying capacity of the bacterial population, which, in a within-host setting, refers to the maximum bacterial load that can be sustained in the human body. $\mu(t)$ denotes the bacterial killing rate induced by antibiotic treatment, which we model as a periodic function. Since we track the bacterial population and antibiotic treatment on a daily basis, we assume $\mu(t)$ to be a 1-periodic function.
\begin{assumption}\label{assumption1}
$\mu(t)$ is a non-negative, continuous, and 1-periodic function.
\end{assumption}
\begin{proposition}\label{prop1}
For any initial value $B(0)\ge 0$, there exists a unique global solution $B(t)$ to \eqref{eq:M1} that is non-negative for all $t\ge 0$.
\end{proposition}

Next, we define the basic reproductive number as 
$$\mathcal{R}_0 = e^{r-\int_0^1 \mu(t)dt},$$
and obtain the following results which are useful for the analysis of the models with resistance development.
\begin{theorem}\label{global_stability}
If $\mathcal{R}_0>1$, then for any initial value $B(0)>0$, system \eqref{eq:M1} has a unique, non-constant, and 1-periodic solution $\Phi(t)$, which is globally asymptotically stable. If $\mathcal{R}_0<1$, then the zero solution $B(t)\equiv 0$ is globally stable.
\end{theorem}

\subsection{Resistance acquired through plasmid transmission}
We consider the case where bacteria acquire resistance through a single mutation gene carried by a plasmid, which is transmitted between bacteria via horizontal gene transfer (HGT). One example of this mechanism is the resistance development in MRSA against vancomycin. HGT has been considered as a nonlinear process that is quite similar to the transmission of infectious diseases. Prior studies modeled this process by either a mass action \cite{DAgata2008,Hall2016} or a mix of density- and frequency- dependent functions \cite{DDR2021}. We here adopt the density- and frequency- dependent functional response transmission rate due to its generalizability. Denote $B_p$ and $B_s$ as the plasmid-bearing and plasmid-free bacterial population, respectively, and $B$ as the total population, we have

\begin{equation}
\begin{split}
\dot B_s(t)&=r_sB_s(t)(1-\frac{B(t)}{K})-\mu_s(t)B_s(t)-\frac{\beta B_s(t)B_p(t)}{a+bB(t)}+\frac{p}{2} r_pB_p(t)\\
\dot B_p(t)&=r_pB_p(t)(1-\frac{B(t)}{K})-\mu_p(t)B_p(t)+\frac{\beta B_s(t)B_p(t)}{a+bB(t)}-\frac{p}{2}r_pB_p(t)
\end{split} \tag{M2} \label{Bsp}
\end{equation}
where $r_s$ and $r_p$ refer to the intrinsic growth rate for each population, and $\mu_s(t)$ and $\mu_p(t)$ are the antibiotic killing rates which we assume to be 1-periodic. $p$ is the probability that a plasmid is lost during the bacterial cell division - when this happens, one of the two daughter cells will become plasmid-free. $\beta$ is the flux rate of HGT.
\begin{notation}\label{rmaxmin}
    $r_{\text{max}}\ =\max\{r_s, r_p\}$ and $r_{\text{min}}\ =\min\{r_s,r_p\}$.
\end{notation}
\begin{assumption}\label{musp}
    $\mu_s(t)$ and $\mu_p(t)$ are non-negative, continuous, and 1-periodic functions.
\end{assumption}
\begin{proposition}\label{prop2}
For any initial value $B_s(0),B_p(0)\ge 0$, \eqref{Bsp} has a unique and bounded global solution that is non-negative for all $t\ge 0$. Furthermore, for any $q>0$, there exists $t_q>0$ such that the solution of \eqref{Bsp} with $t>t_q$ lies in the compact set
\begin{equation}
\mathbb{D}_q \ = \ \{(B_s,B_p)\in \mathbb{R}_+^2:B_s + B_p\le \frac{r_{max}}{r_{min}}\cdot K +q\}.
\end{equation}
\end{proposition}
\begin{proof}
Similar to the proof of Proposition \ref{prop1}, it is easy to show that the right hand side function $f(t,\cdot): \mathbb{R}^2 \to \mathbb{R}^2$ is locally Lipschitz continuous. We thus have the local existence and uniqueness of the solution. By Remark $16.3(f)$ in \cite{AH}, we can show that $\mathbb{R}_+^2$ is an invariant region: for $B_s(0)=0$ and $B_p(0)>0$, we have $\dot B_s(t)=\frac{p}{2}r_pB_p(t)>0$, for $t\in [0,\epsilon)$ with $\epsilon$ being sufficiently small, then $B_s(t)\ge 0$ for $t\in [0,\epsilon)$; similarly, for $B_s(0)>0$ and $B_p(0)=0$, then $\dot B_p(t)\equiv 0$, and we have $B_p(t)\equiv 0$ and $B_s(t)>0$ for $t\in [0,\epsilon)$ with $\epsilon$ being sufficiently small. We therefore obtain the nonnegativity of the solution. The global existence of the solution follows from Theorem 7.6 in \cite{AH} provided with the uniform boundedness shown below.

Let $B= B_p+B_s$, we have
\begin{flalign*}
    \dot B(t) =& r_pB_p(t) - r_p\frac{B_p(t)B(t)}{K} + r_s B_s(t) - r_s\frac{B_s(t)B(t)}{K} - \mu_p(t)B_p(t) - \mu_s(t)B_s(t)\\
    \le & r_pB_p(t)+r_s B_s(t)- r_p\frac{B_p(t)B(t)}{K}-r_s\frac{B_s(t)B(t)}{K}\\
    \le & r_{max} B(t)-r_{min}\frac{B(t)^2}{K} = r_{max}B(t)(1-\frac{B(t)}{\tfrac{r_{max}}{r_{min}}K}).
\end{flalign*}
 Notice that $\hat{B}'(t)=r_{max}\hat{B}(t)(1-\frac{\hat{B}(t)}{\frac{r_{max}}{r_{min}}K})$ is a logistic growth model with carrying capacity $\tfrac{r_{max}}{r_{min}}K$. Thus under the initial condition $\hat{B}(0)=(B_s(0),B_p(0))\in \mathbb{R}_+^2$, for any $q>0$, there exists $t_q>0$ such that $\hat{B}(t)\le \tfrac{r_{max}}{r_{min}}K + q$ for $t>t_q$. By the comparison principle, we obtain that $B(t)\le \tfrac{r_{max}}{r_{min}}K + q$ for $t>t_q$.
\end{proof}

Now we define the basic reproductive numbers and a periodic equilibrium below.
\begin{notation}
    Denote $\mathcal{R}_0^s:= e^{\int_0^1 r_s-\mu_s(t)dt}$, which represents the multiplicative growth factor of plasmid-free offspring over one treatment period. Define $\mathcal{R}_{0}^p := e^{\int_0^1 r_p(1-\frac{p}{2})-\mu_p(t)dt}$, which measures the multiplicative growth factor of plasmid-bearing offspring over one treatment period. 
\end{notation}
\begin{notation}
    Define $$\mathcal{R}_1^s = e^{\int_0^1 r_s(1-2B_s^*(t)/K)-\mu_s(t)dt}$$
    and $$\mathcal{R}_1^p = e^{\int_0^1 \frac{\beta B_s^*(t)}{a+bB_s^*(t)}+(1-\frac{p}{2}-B_s^*(t)/K)r_p-\mu_p(t)dt}$$ where
    \begin{equation}\label{Bs*}
    B_s^*(t)=\frac{Ke^{\int_{0}^t r_s-\mu_s(u)du}}{r_s\int_{0}^te^{\int_{0}^sr_s-\mu_s(u)du}ds+\frac{r_s\int_{0}^1e^{\int_{0}^sr_s-\mu_s(u)du}ds}{e^{\int_{0}^1 r_s-\mu_s(u)du}-1}}.
    \end{equation}
\end{notation}
\begin{theorem}\label{R_0^sp}
(1) The infection-free equilibrium $(0,0)$ is locally asymptotically stable if $\max\{\mathcal{R}_0^s,\mathcal{R}_0^{p}\}<1$.\\
(2) When $\int_{0}^1 r_s-\mu_s(t)dt>0$ and $\max\{\mathcal{R}^s_1,\mathcal{R}^p_1\}<1$, there exists a 1-periodic plasmid-free equilibrium $(0,B_s^*(t))$, which is locally asymptotically stable.
\end{theorem}
\begin{proof}
(1) Linearizing model \eqref{Bsp} at the solution $(0,0)$, we have
$$y'(t)=\begin{pmatrix}
    r_s-\mu_s(t)&\frac{p}{2} r_p\\
    0&r_p(1-\frac{p}{2})-\mu_p(t)
\end{pmatrix}y(t).$$
According to Theorem 4.2.1 in \cite{Farkas}, we calculate the monodromy matrix:
$\begin{pmatrix}
    \mathcal{R}_0^{s}&*\\
    0&\mathcal{R}_0^p \end{pmatrix}$.
We thus obtain the conditions for $(0,0)$ being locally asymptotically stable.\\
\\
(2) We apply the method in \cite{Wang} for the threshold analysis of 
$E_0=(B_s^*(t),0)$. To do so, we regard $B_s$ as the \lq\lq susceptible\rq\rq $ $ compartment, $B_p$ as the \lq\lq infectious\rq\rq $ $ compartment, and plasmid as the transmitting pathogen in the setting of disease transmission. $E_0$ is therefore the infection-free equilibrium (IFE). From the analysis of model \eqref{eq:M1}, we have the explicit form of $E_0$ as in \eqref{Bs*}. We then have the matrices:
$$\mathcal{F}(t,x)=\begin{pmatrix}
    0\\
    \frac{\beta B_pB_s}{a+bB}
    \end{pmatrix},
\mathcal{V}= \mathcal{V}^- - \mathcal{V}^+,$$
where
$$\mathcal{V}^+(t,x)=\begin{pmatrix}
r_sB_s+\frac{p}{2} r_pB_p\\r_pB_p\end{pmatrix},\\
\mathcal{V}^-(t,x)=\begin{pmatrix}r_p r_s B_sB/K+\mu_sB_s+\frac{\beta B_pB_s}{a+bB}\\B_pB/K+\mu_pB_p+\frac{p}{2} r_pB_p
\end{pmatrix}.$$
Thus we have
$F(t)=\frac{\beta B_s^*(t)}{a+bB_s^*(t)}$ and
$V(t)=\mu_p(t)-r_p(1-\frac{p}{2}-B_s^*(t)/K)$.

Before apply Theorems 2.1 and 2.2 in \cite{Wang}, we verify assumptions (A1)-(A7), where (A1)-(A5) can be easily verified based on the expressions of $\mathcal{F}$, $\mathcal{V}^+$, and $\mathcal{V}^-$.

To satisfy (A6), $E_0$ needs to be linearly asymptotically stable in the disease-free subspace. We have $\frac{dy}{dt}=\Bigl(r_s(1-2B_s^*(t)/K)-\mu_s(t)\Bigr)y$. From the prior analysis of model \eqref{eq:M1}, we have $B_s^*$ in \eqref{Bs*} as a locally asymptotic stable 1-periodic solution when $\int_{0}^1 r_s(1-2B_s^*(t)/K)-\mu_s(t)dt<0$.

To meet (A7), we look at $\frac{dy}{dt}=\Bigl(r_p(1-\frac{p}{2}-B_s^*(t)/K)-\mu_p(t)\Bigr)y$ and calculate its monodromy matrix $Y(t,0)=e^{\int_0^t r_p(1-\frac{p}{2}-B_s^*(u)/K)-\mu_p(u)du}$. Thus we require $\phi_{-V}(1)=Y(1,0)=e^{\int_0^1 r_p(1-\frac{p}{2}-B_s^*(t)/K)-\mu_p(t)dt}<1$.

We then calculate the threshold for $(B_s^*(t),0)$ being locally asymptotically stable (the $\mathcal{R}_0$ in Lemma 2.2 \cite{Wang}) and get $\frac{\int_0^1 \frac{\beta B_s^*(t)}{a+bB_s^*(t)}dt}{\int_0^1\mu_p(t)-r_p(1-\frac{p}{2}-B_s^*(t)/K)dt}<1$. Thus $(B_s^*,0)$ is locally asymptotic stable if $r_s>\int_{0}^1 \mu_s(t)dt$ and $\max\{\mathcal{R}^s_1,\mathcal{R}^p_1\}<1$.
\end{proof}
\begin{notation}\label{mumin}
    $\mu_{\text{min}}\ =\min\{\mu_s(t), \mu_p(t)\}$.
\end{notation}
\begin{theorem}
If $\int_0^1r_{\max}-\mu_{\min}\ du<0$, $(0,0)$ is a globally stable solution.
\end{theorem}
\begin{proof}
Let $B= B_p+B_s$. From \ref{Bsp}, $B' = r_pB_p - r_p\frac{B^2}{K} + r_s B_s - r_s\frac{B^2}{K} - \mu_p(t)B_p - \mu_s(t)B_s \le (r_{\max}-\mu_{\min})B$. Thus $M\le e^{\int_0^t r_{\max}-\mu_{\min}\  du}\le e^{t\int_0^1 r_{\max}-\mu_{\min}\  ddu}$. If $\int_0^1r_{\max}-\mu_{\min}\  du<0$, we have that $B\rightarrow 0$ as $t\rightarrow \infty$. Thus $(0,0)$ is a globally stable constant solution of \ref{Bsp}.
\end{proof}

Next, we will examine the uniform persistence under the setting of density-dependent plasmid transformation.
\begin{assumption}\label{abvalue}
Assume $a=1$ and $b=0$.
\end{assumption}
Let $\mathbb{D}_0 = \{(B_s, B_p) \in \mathbb{R}_+^2 : B_p > 0\}$, and let $\partial \mathbb{D}_0 = \mathbb{R}_+^2 \setminus \mathbb{D}_0$. Define the 1-periodic plasmid-free equilibrium as $X_0(t) := (B_s^*(t),0)$, the infection-free equilibrium as $M_1 := \{(0, 0)\}$, and the periodic orbit $M_2 := \{(B_s,0)\in \partial \mathbb{D}_0 : \inf\limits_{t \in [0,1)}B_s^*(t)\le B_s\le \sup\limits_{t \in [0,1)}B_s^*(t)\}$, as well as the Poincar\'e map \( P : \mathbb{R}_+^2 \rightarrow \mathbb{R}_+^2 \) associated with system \eqref{Bsp} by
$$P(z_0) = u(1, z_0)$$
where $z_0 \in \mathbb{R}_+^2 $ and $ u(t, z_0)$ being the unique solution of system \eqref{Bsp} with $z_0$ as the initial condition. Then
$$
P^n(z_0) = u(n, z_0) \quad \text{for } n \geq 0.
$$

\begin{lemma}\label{M_2}
Assume that $\int_{0}^1 r_s-\mu_s(t)dt>0$ and 
$$\int_0^1 \beta B_s^*(t)+r_p(1-\sup_{u\in[0,1)}B_s^*(u)/K-\frac{p}{2})-\mu_p(t)dt>0,$$
then there exists a $\sigma^*>0$ such that for any $z_0 :=(B_s(0),B_p(0))\in \mathbb{D}_0$
we have
$$\lim_{n\rightarrow \infty} \sup \inf\limits_{t\in[0,1)}||P^n(z_0)-X_0(t)||\ge \sigma^*.$$   
\end{lemma}

\begin{proof}
Note that with $\int_{0}^1 r_s-\mu_s(t)dt>0$, we have $\forall B_s(0)>0$,
$$\dot B_s(t)=r_sB_s(t)(1-\frac{B_s(t)}{K})-\mu_s(t)B_s(t)$$ 
has a unique equilibrium $B_s^*(t)$ as in \eqref{Bs*} which is globally attractive in $\mathbb{R}_+$. Also, the following system perturbed by $\sigma>0$
\begin{equation} \label{pertubed_s}
\dot{\hat{B_s}}(t) =r_s\hat{B}_s(t)(1-\frac{\hat{B}_s(t)+\sigma}{K})-\mu_s(t)\hat{B}_s(t)-\beta\sigma\hat{B}_s(t)
\end{equation}
admits a unique solution 
\begin{equation} \label{single_sol_2}
	\hat{B}_s(t,\sigma)=\frac{K\hat{B}_s(0)e^{\int_{0}^t a(s,\sigma)ds}}{r_s\hat{B}_s(0)\int_{0}^te^{\int_{0}^sa(u,\sigma)du}ds+K}, \ \forall \hat{B}_s(0)>0.
\end{equation}
And it admits an unique 1-periodic solution
\begin{equation} \label{Phi2}
	\hat{B}_s^*(t,\sigma)=\frac{K}{r_s}\cdot\frac{e^{\int_{0}^t a(s,\sigma)ds}}{\int_{0}^te^{\int_{0}^s a(u,\sigma) du}ds+\frac{\int_{0}^1e^{\int_{0}^s a(u,\sigma) du}ds}{e^{\int_{0}^1 a(s,\sigma)ds}-1}}
\end{equation}
where $a(t,\sigma):=r_s(1-\frac{\sigma}{K})-\beta\sigma-\mu_s(t)$. 
%$\hat{B}_s^*(0,\sigma)=\frac{K}{r_s}\cdot \frac{e^{\int_{0}^1 a(s)ds}-1}{\int_{0}^1e^{\int_{0}^sa(u)du}ds }$ 
Since $a(t,\sigma)$ is continuous in \(\sigma\), and $\lim_{\sigma \to 0} \int_0^1 a(t,\sigma) dt = \int_0^1 (r_s - \mu_s(t)) dt.$
Then there exists a sufficiently small \(\delta > 0\) such that for all $\sigma\in (0,\delta)$,
\(
\int_0^1 a(t,\sigma) \, dt > 0.
\)
Similar to the proof of Theorem~\ref{global_stability}, for any $\sigma\in (0,\delta)$ we have
\begin{equation}\label{Bstdelta}
\lim_{t\to +\infty}|\hat{B}_s(t,\sigma)-\hat{B}^*_s(t,\sigma)|=0,\,\forall \hat{B}_s(0)>0.
\end{equation}
Applying the implicit function theorem, \(\hat{B}_s^*(0,\sigma)\) depends continuously on \(\sigma\), and
\(
\lim_{\sigma \to 0} \hat{B}_s^*(0,\sigma) = B_s^*(0).
\)
The solution \(\hat{B}_s^*(t,\sigma)\) thus varies continuously with respect to both the initial condition and the parameter. It follows that \(\lim\limits_{\sigma \rightarrow 0} \hat{B}_s^*(t,\ \sigma)=B_s^*(t)\) for all \(t\in [0,1)\), equivalently, \[\lim_{\sigma \rightarrow 0} \sup\limits_{t\in[0,1)}|\hat{B}_s^*(t,\ \sigma)-B_s^*(t)|=0.\] 
Then for any $q_1>0$, there exists $\delta_1>0$ such that for each $\sigma\in(0,\delta_1)$ we have
\begin{equation}\label{lowbd}
\hat{B}_s^*(t,\sigma)>B_s^*(t)-q_1, \ \forall t\in [0,1).
\end{equation}By periodicity of both \(\hat{B}_s^*(t,\sigma)\) and \(B_s^*(t)\), the above inequality extends to all \(t \geq 0\).

Together with \eqref{Bstdelta} we have for a fixed initial condition $\hat{B}_s(0)\in\mathbb{R}_+$ of the perturbed system \eqref{pertubed_s}, for any $q_1>0$ there exists $\delta,\delta_1>0$ so that for any $0<\sigma<\min\{\delta,\delta_1\}$, we have 
$$\hat{B}_s(t,\sigma)\ge B_s^*(t) - 2q_1,$$
for sufficiently large $t$.

%-------------------------------------------------------------------
By the continuous dependence of solutions on the initial point, for \(0<\sigma_1<\min(\delta,\delta_1)\), there exists a \(\sigma^*>0\) such that any \(z_0\in \mathbb{D}_0\) with
$\inf\limits_{t\in[0,1)}||z_0-X_0(t)|| \le \sigma^*$ implies \(
\inf\limits_{s\in[0,1)}||u(t, z_0)-X_0(s)|| < \sigma_1,\, \forall t\in[0,1).
\)
We now prove that 
\[
\lim_{n\rightarrow \infty} \sup \inf\limits_{t\in[0,1)} ||P^n(z_0)-X_0(t)||\ge \sigma^* .
\]    
Assuming for contradiction that \[
\lim_{n\rightarrow \infty} \sup  \inf\limits_{t\in[0,1)}||P^n(z_0)-X_0(t)||< \sigma^* 
\] for some \(z_0 \in \mathbb{D}_0\).
Without loss of generality, we assume \(
 \inf\limits_{s\in[0,1)}||P^n(z_0)-X_0(s)||< \sigma^* 
\)  for some \(z_0 \in \mathbb{D}_0\), \(\forall n>0\). Thus, we have \(
\inf\limits_{s\in[0,1)}||u(t, P^n(z_0))-X_0(s)|| < \sigma_1, \ \forall n\ge 0, \forall t\in[0,1)
\).
Note that any \(t\ge0\) can be expressed as \(t=n+\tilde{t}\) with \(\tilde{t}\in[0,1)\) and \(n\in \mathbb{Z}_+\). Therefore,\[
\inf\limits_{s\in [0,1)}||u(t, z_0)-X_0(s)||=\inf\limits_{s\in[0,1)}||u(\tilde{t},P^n(z_0))-X_0(s)|| < \sigma_1,\, \forall t\ge 0.\]
Substituting \(u(t,z_0)=(B_s(t),B_p(t))\) and \(X_0(t)=(B_s^*(t),0)\), we obtain that 
\[
    B_p(t)<\sigma_1,\,\text{and} \,B_s(t)< \sup_{s\in[0,1)}B_s^*(s) + \sigma_1,\,\ \forall t>0.
\]
Then from the first equation in model \eqref{Bsp}, we obtain 
\begin{equation}
\dot B_s(t)>r_sB_s(t)(1-\frac{B_s(t)+\sigma_1}{K})-\mu_s(t)B_s(t)-\beta \sigma_1B_s(t),\,\forall t\ge 0.
\end{equation}
With \(\sigma_1<\delta\), we have \(\int_0^1 a(t,\sigma_1)dt>0\). By the Comparison Principle and the analysis of the perturbed system \eqref{pertubed_s}, we have 
\begin{equation}\label{range}
    B_s(t) \ge \hat{B}_s(t,\sigma_1)\ge B_s^*(t)-2q_1
\end{equation} for sufficiently large t.
Then we can estimate the second equation in model \eqref{Bsp}:
\begin{align*}
  \dot B_p(t) &> r_pB_p(t)\Bigl(1-\frac{\sup_{s\in[0,1)} B_s^*(s)}{K} - \frac{\sigma_1}{K}\Bigr) - \mu_p(t)B_p(t) \\
  &\quad + \beta \bigl(B_s^*(t) - 2q_1\bigr)B_p(t) - \frac{p}{2} r_p B_p(t) \\
  &= b(t,q_1,\sigma_1)\cdot B_p(t)
\end{align*}
for sufficiently large $t$.
Where we denote $b(t,q_1,\sigma_1):=\beta B_s^*(t)+r_p(1-\frac{\sup_{s\in[0,1)}B_s^*(s)}{K}-\sigma_1/K-\frac{p}{2})-\mu_p(t)-2q_1$. We can choose $q_1>0$ small enough and $\sigma_1<q_1$ so that $\int_0^1 b(t,q_1,\sigma_1)dt>0$. Then the auxiliary equation 
\begin{equation}\label{auxiliaryBp}
\hat{B}_p'(t) = b(t,q_1,\sigma_1)\hat{B}_p(t)
\end{equation}
has solution $\hat{B}_p(t) = \hat{B}_p(0) e^{\int_0^t b(s,q_1,\sigma_1)ds} \to +\infty$ as $t\to +\infty$, $\forall \hat{B}_p(0)>0$. Assume we have $\dot B_p(t)>b(t,q_1,\sigma_1)B_p(t)$ for $t>t_0$ and consider the solution to equation \eqref{auxiliaryBp} with initial condition $\hat{B}_p(t_0)=B_p(t_0)$ we have $\hat{B}_p(t) = B_p(t_0)e^{\int_{t_0}^t b(s,q_1,\sigma_1)ds} \to +\infty$ as $t\to +\infty$. Then by the Comparison Principle we have $B_p(t)\ge \hat{B}_p(t) \to +\infty$ as $t\to +\infty$, which is a contradiction.
\end{proof}

\begin{lemma}\label{M_1}
If $\int_{0}^1 r_s-\mu_s(t)dt>0$, there exists a $\sigma^*>0$ such that for any $z_0=(B_p(0),B_s(0))\in \mathbb{D}_0$ we have \[
\lim_{n\rightarrow \infty} \sup ||P^n(z_0)-(0,0)||\ge \sigma^* .
\]    
\end{lemma}
\begin{proof}
We first consider the equation
\begin{equation}\label{auxiliaryBs}
    \hat{B}_s'(t) = r_s \hat{B}_s(t)(1-\frac{2\sigma}{K}) - \mu_s(t)\hat{B}_s(t)-\beta \sigma \hat{B}_s(t),
\end{equation}
and denote $a(t,\sigma):=r_s(1-\frac{2\sigma}{K})-\mu_s(t)-\beta\sigma$. We note that there exists $\delta>0$ so that $\int_0^1 a(t,\sigma)dt>0$ for all $0<\sigma<\delta$. Thus the solution to \eqref{auxiliaryBs} $\hat{B}_s(t,\sigma)=\hat{B}_s(0)e^{\int_0^t a(t,\sigma)dt}\to +\infty$ for all $\hat{B}_s(0)>0$.

By the continuous dependence of solutions on the initial point, for any \(\sigma_1\in(0,\delta)\), there exists a \(\sigma^*>0\) such that for any \(z_0\in \mathbb{D}_0\) with  \(
||z_0-(0,0)|| \le \sigma^*
\) we have that \(
||u(t, z_0)-(0,0)|| < \sigma_1, \forall \ t\in [0,1).
\)
We now assume for contradiction that 
\[
\lim_{n\rightarrow \infty} \sup ||P^n(z_0)-(0,0)||< \sigma^*
\]    
for some \(z_0 \in \mathbb{D}_0\).
Without loss of generality, we assume that \(
 ||P^n(z_0)-(0,0)||< \sigma^* 
\)  for some \(z_0 \in \mathbb{D}_0\), \(\forall n>0\). Thus, we have \(
||u(t, P^n(z_0))-(0,0)|| < \sigma_1, \ \forall n\ge 0
\) and for \(t\in [0,1)\).
Note that any \(t\ge0\) can be expressed as \(t=n+\tilde{t}\) with \(\tilde{t}\in[0,1)\) and \(n\in \mathbb{Z}_+\). Therefore,\[
||u(t, z_0)-(0,0)||=||u(\tilde{t}, P^n(z_0))-(0,0)|| < \sigma_1 \ \forall \  t>0\]
Substituting \(u(t,z_0)=(B_p(t),B_s(t))\) , we obtain that \(B_p(t)<\sigma_1 \ and \ \ B_s(t)<\sigma_1 \  ,\ \forall t>0.\) Then the first equation in model \eqref{Bsp} can be estimated as
$$\dot B_s(t)>r_s B_s(t)(1-\frac{2\sigma_1}{K})-\mu_s(t)B_s(t)-\beta \sigma_1 B_s(t),\,\forall t>0.$$
Then consider the auxiliary system \eqref{auxiliaryBs} and by Comparison Principle we have $B_s(t)\ge \hat{B}_s(t,\sigma_1)\to +\infty$ as $t\to +\infty$, which is a contradiction.
\end{proof}
\begin{theorem}\label{persistenceBsp}
Assume $\int_{0}^1 \bigl(r_s - \mu_s(t)\bigr)\,dt > 0$ and $\quad \text{and} \quad
\int_0^1 \Bigl(\beta B_s^*(t) + r_p\Bigl(1 - \frac{\sup_{u\in[0,1)} B_s^*(u)}{K} - \tfrac{p}{2}\Bigr) - \mu_p(t)\Bigr)\,dt > 0$,
then there exists a value $\xi > 0$ such that any solution $(B_p(t), B_s(t))$ of system~\eqref{Bsp} 
with the initial value $z_0 \in \mathbb{D}_0$ satisfies $\liminf_{t \to +\infty} B_s(t) \geq \xi
\quad \text{and} \quad
\liminf_{t \to +\infty} B_p(t) \geq \xi$, and the system~\eqref{Bsp} admits at least one positive periodic solution.
\end{theorem}

\begin{proof}
As proved in Proposition~\ref{prop2}, \(\{P^n\} _{n\ge0}\) admits a global attractor $\mathbb{A}$. Let \[\mathbb{M}_{\partial} := \left\{ x \in \partial \mathbb{D}_0 : f^n(x) \in \partial \mathbb{D}_0,\ n \geq 0 \right\} \ = \partial \mathbb{D}_0 ,\]  and \(\mathbb{A}_{\partial}= \mathbb{A} \cap \mathbb{M}_\partial = \partial \mathbb{D}_0 \). The given conditions imply that $(0,0)$ and $(B_s^*(t),0)$ are the only two equilibrium in $\mathbb{A}_{\partial}$ and both unstable. Then $M_1$ and $M_2$ are two isolated invariant sets in $\mathbb{R}_+^2$. Further, similar to the argument in Theorem \ref{global_stability} that any trajectory starts from $\partial\mathbb{D}_0\setminus M_1$ would converge to $M_2$. Thus \(\mathbb{A}_{\partial}\) admits a morse decomposition $\{M_1, M_2\}$. Lemma~\ref{M_2} and Lemma~\ref{M_1} show that 
\(W^s(M_1) \cap \mathbb{D}_0 = \emptyset\) and \(W^s(M_2) \cap \mathbb{D}_0 = \emptyset\). By Theorem 1.3.1 in \cite{Zhao2003}, it follows that $\{P^n\}_{n \geq 0}$ is uniformly persistent with respect to $(\mathbb{D}_0, \partial \mathbb{D}_0)$, and by Theorem 3.1.1 in \cite{Zhao2003}, we have the solutions of system \eqref{Bsp} being uniformly persistent with respect to $(\mathbb{D}_0, \partial \mathbb{D}_0)$. That is, there exists a $\xi > 0$ such that any solution $(B_s(t), B_p(t))$ of system \ref{Bsp} with initial value $(B_s(0), B_p(0)) \in \mathbb{D}_0$ satisfies
$$
\liminf_{t \to \infty} B_s(t) \geq \xi, \quad 
\liminf_{t \to \infty} B_p(t) \geq \xi.
$$
Note that the ultimate boundedness implies the uniform boundedness of solutions for periodic systems of ODEs \cite{Yoshizawa1975}. Thus, the Poinc\'are map $P: \mathbb{R}_+^2 \to \mathbb{R}_+^2$ is compact. By Theorem 1.3.8 in \cite{Zhao2003}, $P$ admits a fixed point $y_0 \in \mathbb{D}_0$, such that $\phi(t, y_0)$ is a periodic solution of \eqref{Bsp}. Due to the uniform persistence and periodicity of the solution we have both components of $\phi(t,y_0)$ being strictly positive for all $t \geq 0$.
\end{proof}
\begin{remark}
    Under the assumption of $p=0$, that is, assuming no plasmid loss during replication, one can obtain a uniform persistence result when $\int_{0}^1 r_s-\mu_s(t)dt>0$ and $\mathcal{R}_0>1$. Under these assumptions, we can derive an upper bound for $B_s(t)$ as $B_s^*(t)>B_s(t)>B_s^*(t)-2q_1$, instead of the estimation in \eqref{range}, and the result follows by similar arguments.
\end{remark}

\subsection{Resistance gained by spontaneous mutation}
We consider the scenario where bacteria acquire resistance through a single spontaneous mutation on the chromosome, which is modeled as a linear rate. An example of such mechanism is the resistance development of MRSA against ciprofloxacin. Denote $B_s$ and $B_m$ as the wild-type and mutant bacterial population, respectively, we have
\begin{equation}
\begin{split}
&\dot B_s(t)=r_sB_s(t)(1-\frac{B(t)}{K})-\mu_s(t)B_s(t)-\varepsilon r_sB_s(t)\\
&\dot B_m(t)=r_mB_m(t)(1-\frac{B(t)}{K})-\mu_m(t)B_m(t)+\varepsilon r_sB_s(t)
\end{split}\tag{M3} \label{Bsm}
\end{equation}
where $r_s$ and $r_m$ refer to the intrinsic growth rate for each population, and $\mu_s(t)$ and $\mu_m(t)$ are the antibiotic killing rates which we assume to be 1-periodic. $\varepsilon$ is the spontaneous mutation probability on the chromosome. In the following, we adopt the Assumption \ref{abvalue} and Notation \ref{rmaxmin}, and a similar assumption to Assumption \ref{musp}.
\begin{notation}\label{rmaxminBsm}
    $r_{\text{max}}\ =\max\{r_s, r_m\}$ and $r_{\text{min}}\ =\min\{r_s,r_m\}$.
\end{notation}
\begin{assumption}\label{musm}
    $\mu_s(t)$ and $\mu_m(t)$ are non-negative, continuous, and 1-periodic functions.
\end{assumption}
\begin{proposition}\label{prop3}
For any initial value $B_s(0),B_m(0)\ge 0$, \eqref{Bsm} has a unique global solution that is non-negative for all $t\ge 0$. Furthermore, for any $q>0$, there exists $t_q>0$ such that the solution of \eqref{Bsm} with $t>t_q$ lies in the compact set
\begin{equation}
\mathbb{D}_q \ = \ \{(B_s,B_m)\in \mathbb{R}_+^2:B\le \frac{r_{max}}{r_{min}}\cdot K +q\}.
\end{equation}
\end{proposition}
\begin{proof}
The proof is almost identical to that of Proposition \ref{prop2}.
\end{proof}
We define the following reproductive numbers and periodic equilibria.
\begin{notation}
    Let $\mathcal{R}_0^{s} = \exp\!\left(\int_0^1 \bigl(r_s(1-\varepsilon) - \mu_s(t)\bigr)\,dt\right)$, 
    which measures the mu\-ltiplicative growth factor of wild-type offspring over one treatment period. 
    Denote 
    $\mathcal{R}_0^{m} = \exp\!\left(\int_0^1 \bigl(r_m - \mu_m(t)\bigr)\,dt\right)$, 
    which measures the multiplicative growth factor of mutant offspring over a treatment period.
\end{notation}

\begin{notation}
    We further have a semitrivial periodic equilibrium $(0, B_{m}^*(t))$ with
\begin{equation}\label{Bm0}
    B_{m}^*(t)=\frac{K}{r_m}\cdot\frac{e^{\int_{0}^t r_m-\mu_m(s)ds}}{\int_{0}^te^{\int_{0}^sr_m-\mu_m(u)du}ds+\frac{\int_{0}^1e^{\int_{0}^s r_m-\mu_m(u)du}ds}{e^{\int_{0}^1 r_m-\mu_m(s)ds}-1}},
\end{equation}
We further denote the stability threshold for $(0, B_{m}^*(t))$ as
\begin{equation}
%\begin{split}
\mathcal{R}_1^{s} = e^{\int_0^1 r_s(1-\varepsilon-B_{m}^*(t)/K)-\mu_s(t)dt}\,\,\text{and}\,\,\mathcal{R}_1^{m}=e^{\int_0^1 r_m-2r_m B_{m}^*(t)/K-\mu_m(t)dt}.
%\end{split}
\end{equation}
\end{notation}
\begin{theorem}\label{theorem_Bsm}
(1) If $\max\{\mathcal{R}_0^{s},\mathcal{R}_0^m\}<1$, then $(0,0)$ is locally asymptotically stable.\\
(2) If $r_m>\int_{0}^1 \mu_m(t)dt$, $\max\{\mathcal{R}_1^{s},\mathcal{R}_1^{m}\}<1$, then the semitrivial 1-periodic solution $(0,B_{m}^*(t))$ is locally asymptotically stable.
\end{theorem}

\begin{proof}
(1) Linearizing the model \eqref{Bsm} at $(0,0)$ we have
$$y'(t)=\begin{pmatrix}
r_s(1-\varepsilon)-\mu_s(t)&0\\
\varepsilon r_s&r_m-\mu_m(t)
\end{pmatrix} y(t).$$
The monodromy matrix is calculated as :
$\begin{pmatrix}
\mathcal{R}_0^{s}&0\\
\mathcal{R}_0^{sm} &\mathcal{R}_0^{m}
\end{pmatrix}$,
with the characteristic multipliers being $\mathcal{R}_0^m$ and $\mathcal{R}_0^{sm}$. The result follows by Theorem $4.2.1$ in \cite{Farkas}.

(2) It is easy to check that the solution $B_{m}^*(t)$ in \eqref{Bm0} is a 1-periodic and non-negative solution to \eqref{Bsm} when $r_m>\int_0^1 \mu_m(t)dt$. The linearized model at $(0,B_{m}^*(t))$ is
$$y'(t)=\begin{pmatrix}
r_s(1-\varepsilon-B_m^*(t)/K)-\mu_s(t) & 0\\
\varepsilon r_s - r_mB_m^*(t)/K & r_m-2r_mB_{m}^*(t)/K-\mu_m(t)\\
\end{pmatrix}y(t).$$
We then have the characteristic multipliers as $\mathcal{R}_1^{s }$ and $\mathcal{R}_1^m$, and the result follows.
\end{proof}
\begin{notation}\label{mumin2}
$\mu_{\text{min}}\ =\min\{\mu_s(t), \mu_m(t)\}$.
\end{notation}
\begin{theorem}
If $\int_0^1r_{\max}-\mu_{\min}\ du<0$, then $(0,0)$ is a globally stable solution.
\end{theorem}

Next, we will examine the uniform persistence of Model \eqref{Bsm}. Let $\mathbb{D}_0 = \{(B_s, B_m\\
)\in \mathbb{R}_+^2 : B_s > 0\}$, and let $\partial \mathbb{D}_0 = \mathbb{R}_+^2 \setminus \mathbb{D}_0$. At the same time, define $M_1=\{(0,0)\}$ and $M_2=\{(0,B_m)\in \partial \mathbb{D}_0: \sup\limits_{t\in[0,1)}B^*_m(t)>B_m>\inf\limits_{t\in[0,1)}B^*_m(t)\}$, as well as the Poincar\'e map \( P: \mathbb{R}_+^2 \rightarrow \mathbb{R}_+^2 \) associated with system \eqref{Bsm} by
\[
P(z_0) = u(1, z_0)
\]
with \( z_0 \in \mathbb{R}_+^2 \), where \( u(t, z_0) \) is the unique solution of system \eqref{Bsm} with
\[
u(0, z_0) = z_0 = \big(B_s(0), B_m(0)\big).
\]
Then
\[
P^n(z_0) = u(n, z_0) \quad \text{for } n \geq 0.
\]

\begin{lemma}\label{M_2_sm}
If $r_m>\int_0^1\mu_m(t)dt$ and $\int_0^1 r_s(1-\sup_{u\in[0,1)}B_m^*(u)/K-\varepsilon)-\mu_s(t)\\
dt>0$, there exists a $\sigma^*>0$ such that for any $z_0=(B_m(0),B_s(0))\in \mathbb{D}_0$ we have that \[
\lim_{n\rightarrow \infty} \sup \inf\limits_{t\in[0,1)}||P^n(z_0)-(B^*_{m}(t),0)||\ge \sigma^* .
\]    
\end{lemma}

\begin{proof}
Consider the auxiliary equation
\begin{equation}\label{auxiliaryBs2}
\hat{B}_s'(t)=a(t,\sigma)\hat{B}_s(t)
\end{equation}
with $a(t,\sigma):=r_s(1-\frac{r_{max}}{r_{min}}-\sigma/K)-\mu_s(t)-\varepsilon r_s$. There exists $\delta>0$ so that for $0<\sigma<\delta$ we have $\int_0^1 a(t,\sigma)dt>0$ and we have the solution to the auxiliary equation $\hat{B}_s(t)=e^{\int_0^1 a(t,\sigma)dt}\hat{B}_s(0)\to +\infty$ as $t\to +\infty$.

By the continuous dependence of solutions on the initial point, for any $\sigma\in(0,\delta)$, there exists a \(\sigma^*>0\) such that any \(z_0\in \mathbb{D}_0\) with  \(
\inf\limits_{s\in[0,1)}||z_0-(0,B_m^*(s))|| \le \sigma^*
\) implies that \(
\inf\limits_{s\in[0,1)}||u(t, z_0)-(0,B_m^*(s)))|| < \sigma, \forall \ t\in [0,1).
\)
We now assume for contradiction that \[
\lim_{n\rightarrow \infty} \sup  \inf\limits_{t\in[0,1)}||P^n(z_0)-(0,B_m^*(t))||< \sigma^* 
\] for some \(z_0 \in \mathbb{D}_0\).
Without loss of generality, we assume that \(
 \inf\limits_{s\in[0,1)}||P^n(z_0)-(0,B_m^*(t))||< \sigma^* 
\)  for some \(z_0 \in \mathbb{D}_0\), \(\forall n>0\). Thus, we have \(
\inf\limits_{s\in[0,1)}||u(t, P^n(z_0))-(0,B_m^*(s))|| < \sigma, \ \forall n\ge 0
\) and for \(t\in [0,1)\).
Note that any \(t\ge0\) can be expressed as \(t=n+\tilde{t}\) with \(\tilde{t}\in[0,1)\) and $n\in\mathbb{Z}_+$. Therefore,\[
\inf\limits_{s\in[0,1)}||u(t, z_0)-(0,B_m^*(s))||=\inf\limits_{s\in[0,1)}||u(\tilde{t}, P^n(z_0))-(0,B_m^*(s))|| < \sigma \ \forall \  t>0\]
Substituting \(u(t,z_0)=(B_s(t),B_m(t))\) , we obtain that 
\[
    B_s(t)<\sigma\,\,\text{and}\,\,B_m(t)<\sup_{u\in[0,1)}B_m^*(u)+\sigma,\, \forall t>0.
\]
Then the first equation in model \eqref{Bsm} can be estimated as
\[
\dot B_s(t)>r_sB_s(t)(1-\sup_{u\in[0,1)}B_m^*(u)/K-\sigma/K)-\mu_s(t)B_s(t)-\varepsilon r_sB_s(t),
\]
for $t>0$. By similar arguments as in Lemma \ref{M_2} we have $B_s(t)\to +\infty$ as $t\to +\infty$, which is a contradiction.
\end{proof}

\begin{lemma}\label{M_1_sm}
If $\int_0^1 r_m-\mu_m(t)dt>0$, then there exists a $\sigma^*>0$ such that for any $z_0=(B_m(0),B_s(0))\in \mathbb{D}_0$ we have that \[
\lim_{n\rightarrow \infty} \sup ||P^n(z_0)-(0,0)||\ge \sigma^* .
\]    
\end{lemma}

\begin{proof}
For the perturbed second equation in model \eqref{Bsm}, we have
\begin{equation} 
    \dot{\hat{B}}_m(t)=b(t,\sigma)\cdot \hat{B}_m(t)\label{perturbed_sm_m},
\end{equation}
where \(b(t,\sigma)=r_m(1-\frac{2\sigma}{K})-\mu_m(t)\). There exists a $\delta>0$ so that for any $\sigma\in(0,\delta)$ we have $\int_0^1 b(t,\sigma)dt>0$. Then the solution to \eqref{perturbed_sm_m} $\hat{B}_m(t)=e^{\int_0^t b(s,\sigma)ds}\hat{B}_m(0)\to +\infty$ as $t\to +\infty$.

By the continuous dependence of solutions on the initial point, for such $\sigma$, there exists a \(\sigma^*>0\) such that any \(z_0\in \mathbb{D}_0\) with  \(
||z_0-(0,0)|| \le \sigma^*
\) implies that \(
||u(t, z_0)-(0,0)|| < \sigma, \forall \ t\in [0,1).
\)
We now assume for contradiction that \(
\lim_{n\rightarrow \infty} \sup  ||P^n(z_0)-(0,0)||< \sigma^* 
\) for some \(z_0 \in \mathbb{D}_0\).
Without loss of generality, we assume that \(
 ||P^n(z_0)-(0,0)||< \sigma^* 
\)  for some \(z_0 \in \mathbb{D}_0\), \(\forall n>0\). Thus, we have \(
||u(t, P^n(z_0))-(0,0)|| < \sigma, \ \forall n\ge 0
\) and for \(t\in [0,1)\).
Similar to the proof of Lemma \ref{M_2_sm}, we have \[
||u(t, z_0)-(0,0)||< \sigma \ \forall \  t>0.\]
And we obtain that \(B_m(t)<\sigma \ and \ \ B_s(t)<\sigma \  ,\ \forall t>0.\)

Estimating the second equation in model \eqref{Bsm}, we have
\[
\dot B_m(t)>r_mB_m(t)(1-\frac{2\sigma}{K})-\mu_s(t)B_s(t),
\]
By similar arguments as in Lemma \ref{M_2} we have $B_m(t)\to +\infty$ as $t\to +\infty$, which is a contradiction. 
\end{proof}

\begin{theorem}\label{persistenceBsm}
If $\int_{0}^1 r_m-\mu_m(t)dt>0$ and $\int_0^1 r_s(1-\sup_{u\in[0,1)}B_m^*(u)/K-\varepsilon)-\mu_s(t) dt>0$, then there exists a value $\xi > 0$ such that any solution $(B_s(t), B_m(t))$ of system \eqref{Bsm} with the initial value $z_0 \in \mathbb{D}_0$ satisfies \(\liminf\limits_{t \to +\infty} B_m(t) \geq \xi, \ \liminf\limits_{t \to +\infty} B_s(t) \geq \xi,\) and system \eqref{Bsm} admits at least one positive periodic solution.
\end{theorem}
\begin{proof}
The proof is based on Lemma \ref{M_1_sm} and Lemma \ref{M_2_sm} and similar to that in Theorem \ref{persistenceBsp}.
\end{proof}

\begin{remark}\label{persistencecomparison}
    Comparing the conditions for uniform persistence in systems \eqref{Bsp} and \eqref{Bsm}, we find that, under similar settings for the resistant strains (that is, assume $r_p=r_m$ and $\mu_p(\cdot)=\mu_m(\cdot)$), system \eqref{Bsp} exhibits uniform persistence over a wider range of parameters. Thus, resistance acquired through plasmid transmission is a mechanism more likely to promote the persistence of antibiotic resistance.
\end{remark}

\section{Parameterization}
We parameterize the time-dependent bacterial killing rate function, $\mu_i(t),\, i = s,p,m$, as the composition of two components: the time-dependent antibiotic concentration function, $C(t)$, and the concentration-dependent killing rate function, $\mu_i(C),\, i = s,p,m$. Parameter values are based on MRSA bacteria and the antibiotic moxifloxacin, and are summarized in Table \ref{table}.

\subsection{Concentration-dependent killing rates}
We adopt two key pharmacokinetic (PK) concepts to quantify the killing rate of bacterial cells under varying antibiotic concentrations. The minimum inhibitory concentration (MIC) is defined as the lowest antibiotic concentration that inhibits the growth of a susceptible bacterial population, while the mutant prevention concentration (MPC) is the lowest antibiotic concentration that prevents the growth of resistant mutants \cite{Drlica2007,Drlica2006,Metzler2004,Zhao2001}. For any given bacterial isolate and antibiotic, both MIC and MPC can be experimentally determined, with MPC typically greater than MIC. The interval $[\text{MIC}, \text{MPC}]$ defines the so-called mutant selection window (MSW) - a “dangerous zone” in which antibiotic concentrations may suppress susceptible bacteria while simultaneously favoring the selection and growth of resistant subpopulations.

For each bacterial population type, the antibiotic killing rate depends on the actual drug concentration. Various pharmacodynamic (PD) models have been developed to capture this concentration-dependent killing effect, with MIC and MPC values serving as key parameters in these functions. A comprehensive summary of such models is provided by Imran and Smith \cite{Imran2006}. In this study, we illustrate this approach using the Hill function as an example:
\begin{equation}
\mu_i(C):=\frac{\mu_i^{\text{max}}\cdot(\frac{C}{\text{MIC}_i})^{k_i}}{(\frac{C}{\text{MIC}_i})^{k_i}-1+\frac{\mu_i^{\text{max}}}{r_i}},\,i=s,p,m.
\end{equation}
As an example, in the expression for $\mu_s$, the MIC corresponds to the minimum inhibitory concentration of the susceptible bacteria ($B_s$). Here, $r_s$ denotes the intrinsic growth rate, $\mu_s^{\text{max}}(>r_s)$ is the maximal drug-induced killing rate, and $k_s$ is the Hill function shape parameter, which we set to $k_s=1$ for simplicity. The function $\mu_s(C)$ satisfies the following properties:
\begin{itemize}
    \item $\mu_s(\text{MIC})=r_s$, reflecting the definition of MIC as the concentration threshold where drug killing balances bacterial growth;
    \item $\mu_s(0)=0$, indicating zero killing in the absence of antibiotics;
    \item $\lim_{C\to+\infty}\mu_s(C)=\mu_s^{\text{max}}$, representing the maximal killing achievable at high concentrations;
    \item $\mu_s'(C)>0$, ensuring that the killing rate increases monotonically with concentration.
\end{itemize}
Bacterial populations in $B_p$ and $B_m$ are antibiotic resistant, with killing rates intersect their respective net growth rates at higher drug concentrations ($\text{MIC}_i,\,i=p,m$). In particular, it is natural to assume that $\text{MIC}_i$ for $i=p,m$ lies within the mutant selection window $[\text{MIC},\text{MPC}]$.

\subsection{Time-dependent antibiotic concentration}
For a precise adherence to the treatment schedule, we model the in-host antibiotic concentration $C(t)$ as a periodic function resulting from repeated dosing:
{\small
\begin{equation}\label{Concentration}
C(t)=\frac{1}{2}(C_{\text{max}}+C_{\text{min}})+\frac{1}{2}(C_{\text{max}}-C_{\text{min}})\sin(\frac{2\pi}{p}t-\frac{\pi}{2}).
\end{equation}}
Within each dosing cycle of length $p$, the concentration oscillates between the maximal level $C_{\text{max}}$ and the minimal level $C_{\text{min}}$. The period $p$ corresponds to the treatment duration between two successive doses. The area under the curve (AUC) value over a full day is then defined as 
\begin{equation}\label{AUC}
\text{AUC} = \int_0^1 C(t)dt.
\end{equation}
We investigate the effects of once-daily dosing ($p=1$) and twice-daily dosing ($p=\tfrac{1}{2}$), under the assumption that both regimens deliver the same total daily amount of antibiotic. In the twice-daily protocol, the total daily dose is evenly divided between two administrations. For the once-daily scenario, we adopt the reported $C_{\text{max}}$ value from the moxifloxacin drug label and calibrate $C_{\text{min}}$ to match the specified AUC value \cite{FDA}. For the twice-daily scenario, we assume $C_{\text{max}}$ values that are proportionally reduced relative to the once-daily case and parameterize $C_{\text{min}}$ accordingly, ensuring that the AUC is preserved across dosing strategies.

\subsection{Bang-bang type killing rate}
We aim to investigate strategies for managing missed doses, with the goal of identifying optimal responses when a scheduled administration is not taken as planned. Two intuitive remedial approaches are considered: (a) a {\bf catch-up} strategy, in which the missed dose is taken immediately upon recognizing the delay, and (b) a {\bf compensatory double-dose} strategy, in which a double dose is administered at the next scheduled dosing time. Modeling these scenarios using traditional PK/PD frameworks would require complex assumptions about how drug concentrations evolve under each remedial action. To enable a more tractable analysis and extract qualitative insights, we adopt a simplified framework based on a bang-bang control function, which approximates dosing behavior as a binary on/off regime.

We define the efficacy of each antibiotic dose, as reflected in the killing rate $\mu_i(t)$ ($i = s,p,m$) for each bacterial compartment, using a piecewise-defined function:
$$
\mu_i(t) = 
\begin{cases}
\mu_i^{\mathrm{on}}, & t \in [0, \tau) \\
0,       & t \in [\tau, 1)
\end{cases},
$$
where $\mu_i^{\mathrm{on}}$ denotes the constant killing rate during the drug effective period, and $\tau$ specifies the duration of this effect within a normalized dosing cycle. The interval $[0,1)$ represents one full dosing cycle, scaled to unit length for convenience.

For simplicity, we assume that the drug effective period is shorter than the duration of each dosing cycle ({\it i.e.} we assume short-half-life drugs). A straightforward calculation shows that, in order for the bacteria-free equilibrium to remain stable, the cumulative killing effect over a full 24-hour period must exceed the bacterial growth potential - a condition equivalent to requiring that the AUC be sufficiently large. To illustrate, we examine a twice-daily regimen in which doses are scheduled at 12:00 a.m. and 12:00 p.m., with each administration conferring an effective killing duration of six hours. In this setup, we further assume that any missed dose occurs at 12:00 a.m. and is recognized only at 6:00 a.m.

In the simulations, model parameters are selected such that the bacteria-free equilibrium remains stable under full adherence to the dosing schedule but becomes unstable when half of the scheduled doses are missed. To ensure a fair comparison between the two models, we assign equivalent MIC values to the resistant strains in both cases.

\section{Simulation Results}
All simulations are based on a daily oral dose of 400 mg moxifloxacin. For the once-daily regimen, we set $C_{\text{max}} = 3$ mg/L and use an AUC of $36.1$ mg$\cdot$h/L, as reported in the moxifloxacin drug label \cite{FDA}. The corresponding minimal concentration, $C_{\text{min}} = 0.33$ mg/L, is calculated to be consistent with the specified AUC. For the twice-daily regimen, we assume that the maximal concentration $C_{\text{max}}$ is proportional to that of the once-daily regimen, with the proportionality factor chosen to be no less than 56\% - the minimum ratio that guarantees a non-negative minimal concentration $C_{\text{min}}$.

\subsection{MSW, $C_{\text{max}}$, and treatment efficacy}
%Comparing the $\mathcal{R}_0$ values between once-daily and twice-daily regimens for models (M1)–(M5), we observe that the twice-daily regimen consistently outperforms the once-daily regimen by enabling clearance across a broader range of MIC values (e.g., Figure \ref{fig:1}$(a)$).
PK/PD principles emphasize that the width of the mutant selection window (MSW) strongly influences treatment efficacy: regimens associated with wider MSWs can suppress a broader range of MIC values. Under the assumption of a fixed AUC across protocols, and using the antibiotic concentration model \eqref{Concentration}, lower $C_{\text{max}}$ values correspond to wider MSWs (Figure \ref{fig:illustration}).

In our framework, treatment efficacy is evaluated using the threshold MIC value, defined as the MIC level at which the basic reproductive number satisfies $\mathcal{R}_0 = 1$. At or below this threshold, the treatment protocol is sufficient to eradicate the bacterial population. Consequently, a higher threshold MIC value reflects a more effective regimen.

For the model without genetic alteration (M1), the threshold MIC under the twice-daily regimen decreases as the maximal antibiotic concentration $C_{\text{max}}$ increases (Figure \ref{fig:1}$(a)$). Thus, in this simple setting, the PK/PD concept of MSW width is consistent with the threshold MIC–based measure of treatment efficacy. However, in more complex models such as (M2) and (M3), this relationship no longer holds. As shown in Figure \ref{fig:2}, the threshold MIC does not decrease monotonically with increasing $C_{\text{max}}$. This indicates that the MSW concept alone may be insufficient for evaluating treatment efficacy. The quantitative behavior of the killing rate throughout a treatment plays a more direct role in shaping bacterial dynamics, particularly in systems where populations are structured by multiple resistance levels.

\subsection{Impacts of resistance development mechanism}
We investigate the impacts of resistance development mechanism on two aspects: infection clearance and promotion of resistant genes.

{\bf Infection clearance.}
Figure \ref{fig:2}(a) compares the threshold MIC values for resistant bacterial populations under different resistance development mechanisms. Our analysis shows that antibiotic treatments are more effective at clearing infections when resistance arises through plasmid transmission. These findings provide a theoretical rationale for the use of fourth-generation fluoroquinolones, which should further delay resistance emergence, since MRSA would require both a plasmid and a chromosomal mutation to achieve full resistance.

{\bf Promotion of resistant genes.}
For the one-step resistance development mechanism, as illustrated by the simulation in Figure \ref{fig:2}(b), both pathways facilitate the expansion of resistant bacterial strains, and clinical symptoms fail to improve following treatment initiation. However, the mutation-induced resistance mechanism more strongly accelerates the growth of resistant bacteria, allowing them to dominate the population soon after treatment begins. %Consequently, in cases of failed treatment, ciprofloxacin may promote colonization and long-term infection by antibiotic-resistant strains.

\subsection{Protocol for missed doses}
We use bang-bang type functions to investigate how to deal with a missed dose during a 10-day antibiotic therapy, assuming short half-life drugs. We compared the effects of the catch-up strategy and the compensatory double-dose strategy on the total bacterial load at the end of the treatment, relative to a standard twice-daily schedule.

Our simulations show that, regardless of the resistance development mechanism, a catch-up dose would mitigate the harm of a missed dose. For infections with bacteria having low MIC values, the catch-up strategy would yield almost the same outcome as the baseline schedule. For infections triggered by bacteria with high MIC values, the catch-up strategy would only increase the total final bacterial load by at most 4\%. Additionally, the catch-up strategy would work slightly better than the baseline strategy if performed within the first few days of treatment.

On the other hand, compensatory double dosing is not recommended under any conditions. Furthermore, if the bacteria acquire resistance via spontaneous mutation, the compensatory double strategy would result in a significantly higher bacterial load than the other mechanisms. 

%This stresses that compensatory double dosing is especially not recommended for ciprofloxacin, a commonly used broad-spectrum antibiotic.

These findings suggest that the catch-up strategy is an overall recommendation for a missed dose, and a missed dose should be especially avoided when nearing the end of the treatment. These findings align with the general advice in pharmaceutics for short half-life drugs, whereas double dosing may be a better choice for long half-life drugs \cite{Counterman2021}.

\section{Discussion}
In this study, we used mathematical models to demonstrate that bacterial resistance mechanisms against antibiotics can substantially influence treatment outcomes. Theoretically, we found that treatment success thresholds differ depending on the underlying mechanism, as do the conditions for bacterial persistence. Our numerical simulations further showed that, under identical pharmacodynamic conditions, plasmid-mediated resistance is associated with higher rates of treatment success and is less likely to select for resistant strains during failed treatments. This observation aligns with clinical evidence, such as the higher population-level resistance rate to vancomycin compared with ciprofloxacin, despite vancomycin having been developed decades earlier. These findings suggest that differences in resistance development mechanisms may play a crucial role in shaping resistance prevalence.

We also showed that a twice-daily dosing schedule consistently outperforms a once-daily regimen for both resistance mechanisms, and that a catch-up dosing strategy is preferable when doses are missed. In contrast, the compensatory double-dose strategy leads to treatment failure, with particularly poor outcomes in the case of mutation-mediated resistance.

A limitation of this study is that we did not incorporate immune system dynamics, which play an essential but complex role in infection clearance. As an approximation, one could assume that an equilibrium is bacteria-free if the total bacterial population falls below $10^2,\text{CFU}/\text{mL}$, which is the detection limit for bacterial load and a reasonable threshold at which the immune system can clear infection.

Our parameterization was based on the fourth-generation fluoroquinolone moxifloxacin, for which bacteria must acquire two consecutive resistance determinants to achieve full resistance. To model this class of fluoroquinolones more generally, one should consider scenarios where bacteria require two consecutive mutations. For example, in MRSA infections treated with moxifloxacin, resistance typically requires both a plasmid-borne determinant and a chromosomal mutation. We refer to the constant-coefficient ODE model investigated in \cite{DDR2021} as a foundation for this framework.

Although our simulations considered a simplified resistance scenario, they capture realistic features of MRSA treatment with moxifloxacin. For a given MRSA isolate, microbiologists often presume that either the plasmid or chromosomal resistance gene is pre-existing, while spontaneous mutations are unlikely to occur within infections of low bacterial load. In this context, model \eqref{Bsp} represents the case where a plasmid gene is pre-existing, corresponding to higher $\text{MIC}_i\,(i=s,p)$ values for the chromosomal gene when present, and lower values when absent. In contrast, model \eqref{Bsm} reflects the situation where no plasmid gene is initially present and a large bacterial load increases the probability that spontaneous mutation will occur during infection.

\section*{Acknowledgements}
This research was partially supported by the National Science Foundation (DMS-2052648).

\section*{Conflict of Interest}
The authors declare no conflict of interest.

\newpage

\begin{table}
\small
\begin{threeparttable}
    \centering
    \caption{List of Parameters and their Values}
    \begin{tabular}{|l|p{5cm}|l|l|} \hline 
    Symbol & Definition & Value & Ref.\\ 
    \hline
    $r$ & intrinsic growth rate & $2.7726^{(i)}$/day & \cite{DAgata2008} \\  
    \hline
    $p$ & plasmid loss prob. during division & 0.4 & \cite{DAgata2008}\\
    \hline
    $K$ & carrying capacity & 1  & \\  
    \hline
    $\beta$ & HGT flux rate & $10^{-13}/\text{day}$ & \cite{Lopatkin2016} \\  
    \hline
    $a$ & HGT dependence on donor density & $0.5$ & \cite{DDR2021} \\ 
    \hline
    $b$ & HGT dependence on donor frequency & $0.5$ & \cite{DDR2021}\\ 
    \hline
    $\varepsilon$ & mutation rate & $10^{-7}$ & \cite{Szafrańska2019}\\
    \hline
    $C_{\max}$ & max serum moxifloxacin conc. & $3^{(ii)}$ mg/L & \cite{FDA}\\ 
    \hline
    $C_{\min}$ & min serum moxifloxacin conc. & $0.33^{(iii)}$ mg/L & \cite{FDA} \\ 
    \hline
    $\mu_{\max}$ & max drug killing rate & $8.7726^{(iv)}$/day & \cite{Lemaire2011}\\ 
    \hline
    $\text{MIC}_i$ & min inhibitory conc. & Varied$^{(v)}$ & \cite{Andrews1999}\\
    \hline
    $k$ & Hill parameter & 1 & Text \\ 
    \hline
    \end{tabular}
    \begin{tablenotes}
        \item[$(i)$] Uniform growth rate for all bacterial populations.
        \item[$(ii)$] $C_{\max}$ refers to once-daily serum moxifloxacin conc. (single oral dose) \cite{FDA}.
        \item[$(iii)$] $C_{\min}$ refers to once-daily serum conc., calculated from AUC in moxifloxacin facts \cite{FDA}.
        \item[$(iv)$] $\mu_{\max}$ estimated assuming $2$–$3\log_{10}$ cfu/mg reduction within 24h \cite{Lemaire2011}.
        \item[$(v)$] Susceptible strains: MIC $0.125\!\sim\!1$; resistant: MIC $2\!\sim\!4$.
    \end{tablenotes}
    \label{table}
\end{threeparttable}
\end{table}

\begin{figure}
    \centering
    \includegraphics[width=0.75\linewidth]{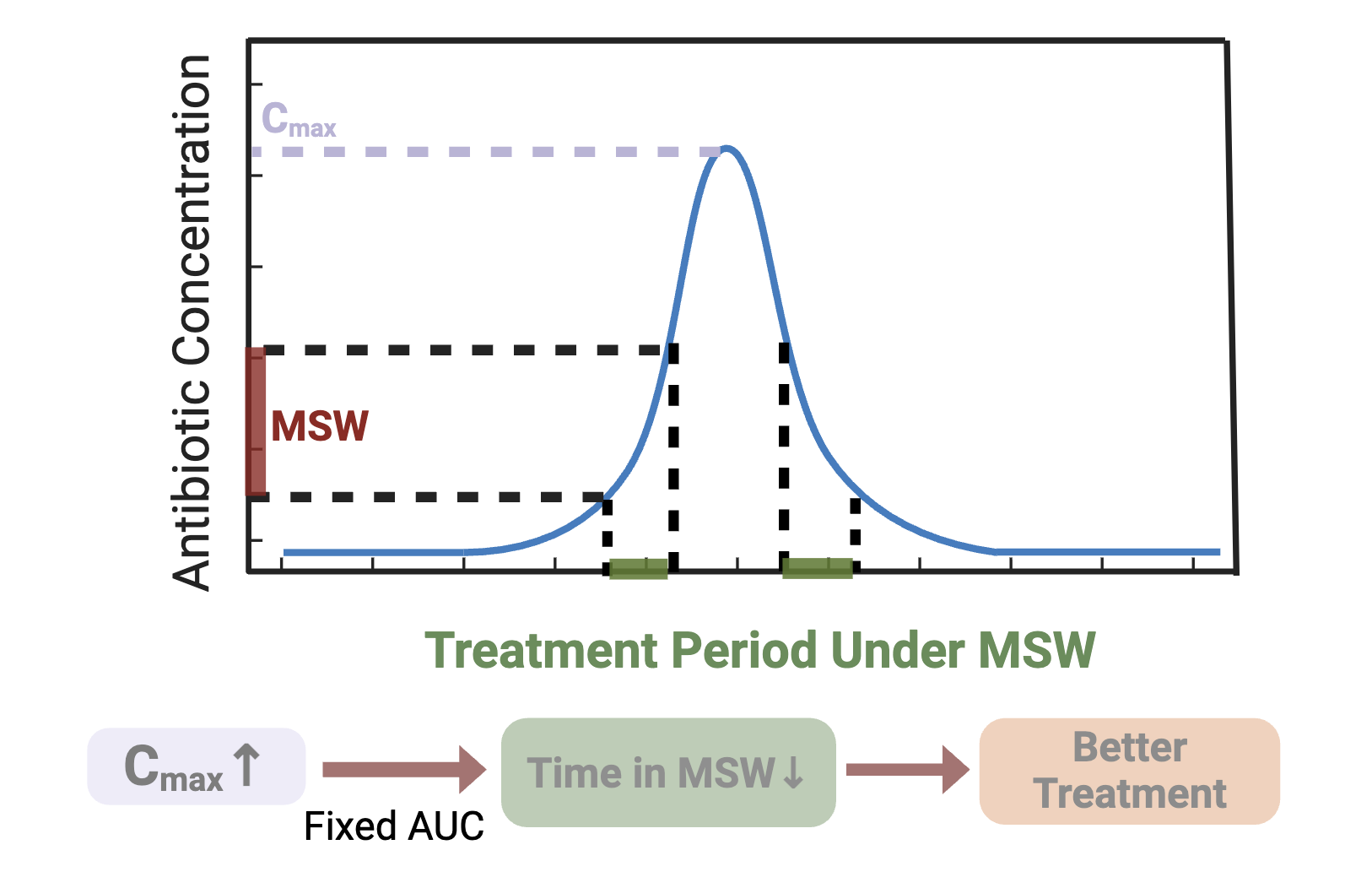}
    \caption{Relationship between MSW and $C_{\text{max}}$ under fixed AUC. A higher $C_{\text{max}}$ sharpens the concentration–time profile and shortens exposure within the MSW, minimizing resistance risk; lower $C_{\text{max}}$ prolongs MSW exposure.}
    \label{fig:illustration}
\end{figure}

\begin{figure}
\begin{subfigure}{0.48\textwidth}
    \centering
    \includegraphics[width=0.99\linewidth]{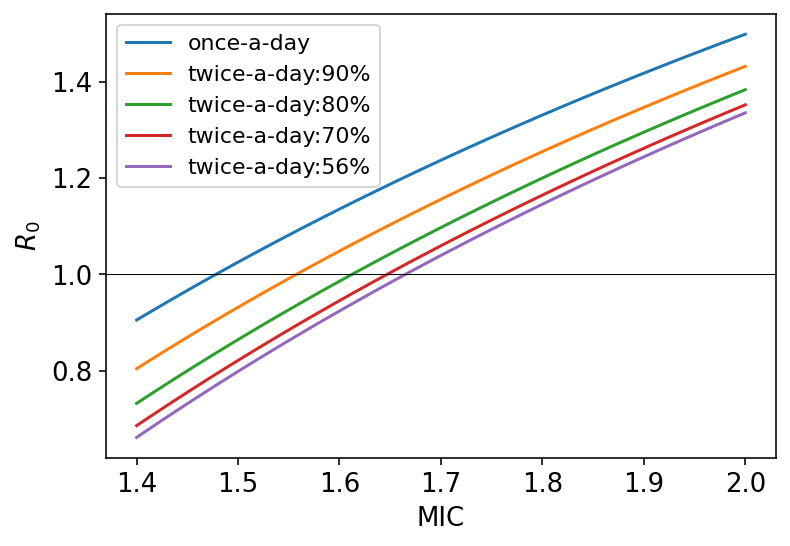}
    \caption{$R_0$ as a function of MIC}
\end{subfigure}
\begin{subfigure}{0.48\textwidth}
    \centering
    \includegraphics[width=0.99\linewidth]{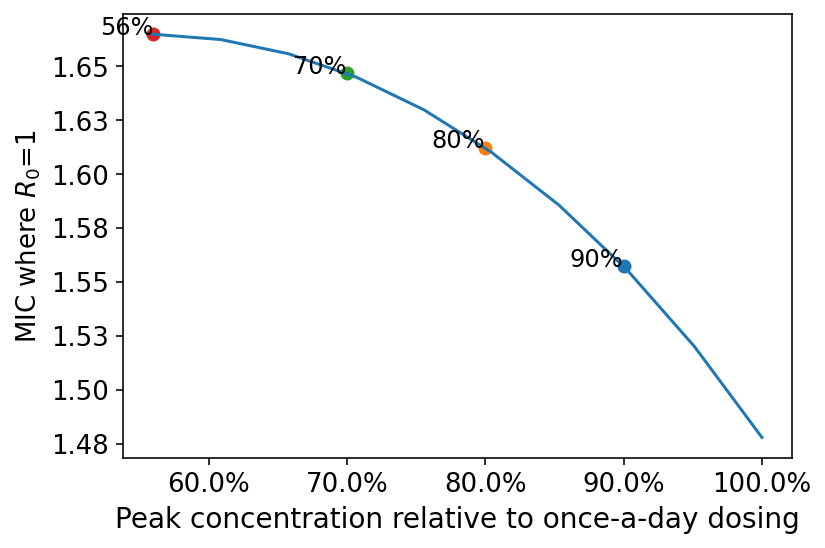}
    \caption{Maximal MIC for infection clearance}
\end{subfigure}
\caption{Simulation of Model (M1). $(a)$ The basic reproductive number $\mathcal{R}_0$ is plotted as a function of the MIC for once-daily and twice-daily regimens. In the figure, it is assumed that the maximal antibiotic concentration for the twice-daily strategy achieves 56\%, 70\%, 80\%, and 90\% of the maximal concentration for the once-daily strategy. $(b)$ Threshold MIC value for $\mathcal{R}_0=1$ under various possibilities of the maximal concentration for the twice-daily regimen. The horizontal axis shows the ratio of $C{\text{max}}$ for the twice-daily regimen relative to the once-daily regimen.}
\label{fig:1}
\end{figure}

\begin{figure}
\begin{subfigure}{0.48\textwidth}
    \includegraphics[width=0.99\linewidth]{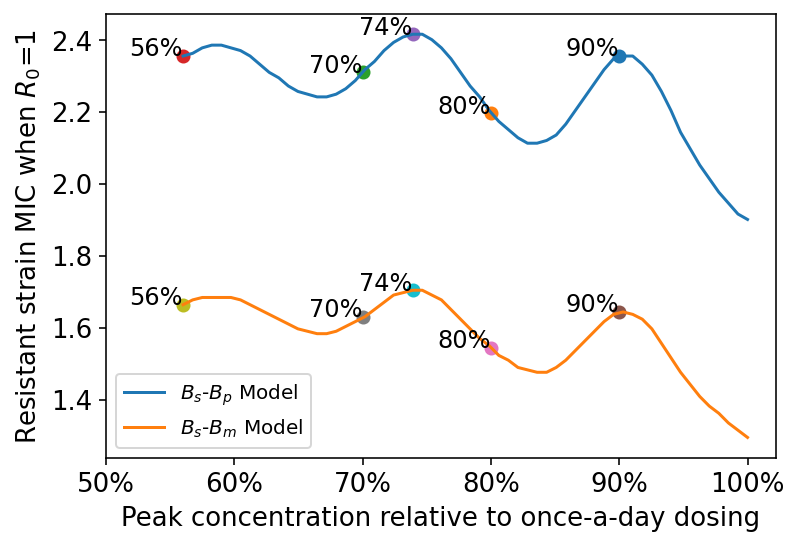}
    \caption{}
\end{subfigure}
\begin{subfigure}{0.48\textwidth}
    \includegraphics[width=0.99\linewidth]{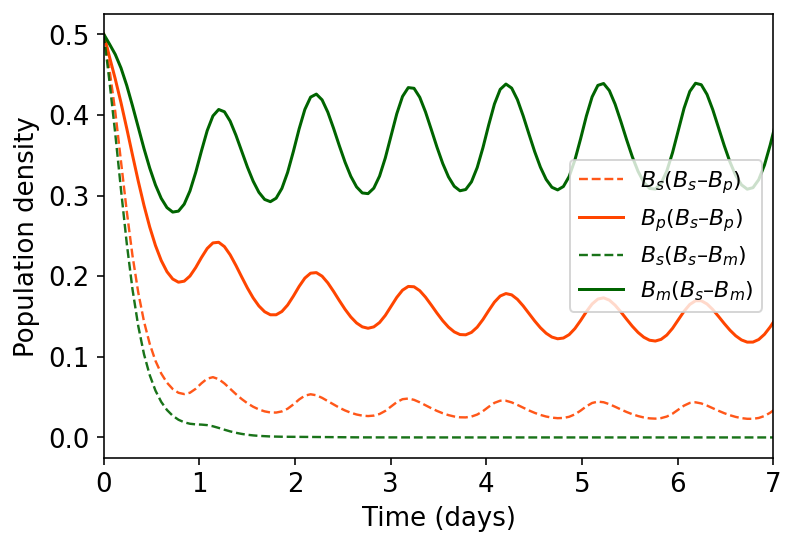}
    \caption{}
\end{subfigure}
\caption{Simulations for Models (M2) and (M3). (a) Threshold MIC for Infection Clearance under Different Resistance Development Mechanisms. Simulations are based on models (M2) and (M3), where the threshold MIC values for the resistant bacterial population ($\text{MIC}_p$ and $\text{MIC}_m$, respectively) were calculated such that $\mathcal{R}_0 = 1$ in each model. The horizontal axis shows the ratio of $C{\text{max}}$ for the twice-daily regimen relative to the once-daily regimen. (b) In both models, we set $\text{MIC}_s = 0.5$ and assign identical resistant strain values, $\text{MIC}_p = \text{MIC}_m = 3$.}
\label{fig:2}
\end{figure}

\begin{figure}
\begin{subfigure}{0.48\textwidth}
    \includegraphics[width=0.99\linewidth]{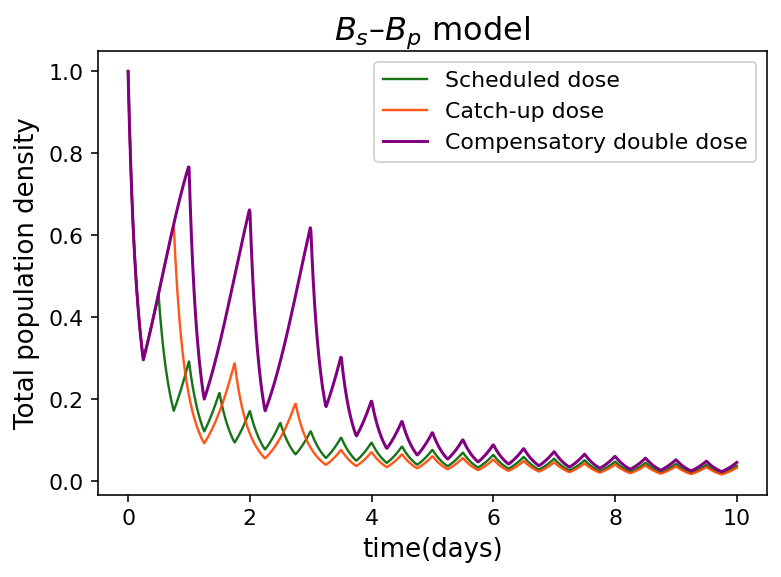}
    \caption{Missed Dose in the First Three Days}
\end{subfigure}
\begin{subfigure}{0.48\textwidth}
    \includegraphics[width=0.99\linewidth]{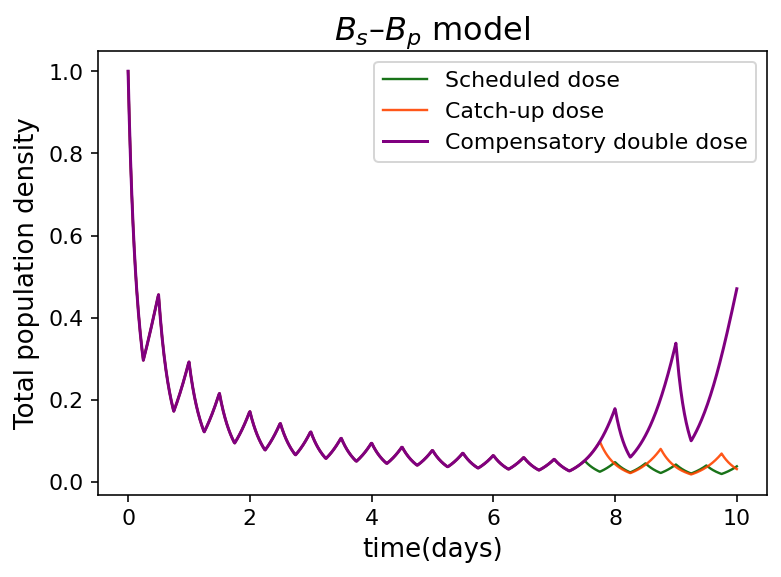}
    \caption{Missed Dose in the Last Three Days}
\end{subfigure}
\caption{The Impact of Make-Up Protocols on Total Bacterial Load Following Missed Doses. The figure illustrates the total bacterial load over time for different missed-dose protocols in a twice-daily dosing regimen, as simulated by Model \eqref{Bsp}. The missed dose is realized midway between consecutive scheduled doses. Parameters were specifically chosen to demonstrate that a twice-daily schedule effectively clears the bacterial infection, whereas a once-daily regimen with the same dose amount at each intake fails to do so.}
\end{figure}
     
\clearpage         
\phantomsection    
\bibliographystyle{siamplain}

\end{document}

\maketitle

\section{Proof of Proposition 2.1}
\begin{proof}
By Theorem 7.4 in \cite{AH}, the local existence and uniqueness of the solution is guaranteed if $f(t,B) = rB(1-B/K)-\mu(t)B$ is continuous for $t\in \mathbb{R}_+$ and Lipschitz continuous for $B\in \mathbb{R}_+$. The continuity of $f(\cdot, B)$ is obvious under Assumption 2.1. For any $t\in \mathbb{R}_+$, we have
\begin{equation*}
\left|f(t,B_1)-f(t,B_2)\right| \le \bigl(r+\left|B_1+B_2\right|+\max_{0\le t<1}\mu(t)\bigr)\cdot \left|B_1-B_2\right|,
\end{equation*}
hence the local Lipschitz continuity of $f(t,\cdot)$.

By Remark $16.3(f)$ in \cite{AH}, the solution $B(t)$ with $B(0)\ge 0$ will stay non-negative since $B(t)\equiv 0$ is the unique global solution with $B(0)=0$.

In (M1), we have that $\dot B \le rB(1-\frac{B}{K})$. By the Comparison Principle we have $B(t)\le \max\{K,B(0)\}$. We thus have the global existence and uniqueness of the solution following Theorem 7.6 in \cite{AH}. 
\end{proof}

\section{Proof of Theorem 2.1}

\begin{proof}
It is easy to verify that
\begin{equation} \label{single_sol}
	\varphi(t,B(0))=\frac{KB(0)e^{\int_{0}^t r-\mu(s)ds}}{rB(0)\int_{0}^te^{\int_{0}^sr-\mu(u)du}ds+K}
\end{equation}
is a solution to (M1).
Let 
$$\varphi(1,B_0^*)=\frac{K B_0^* e^{\int_{0}^1 r-\mu(s)ds}}{rB_0^* \int_{0}^1e^{\int_{0}^sr-\mu(u)du}ds+K}=B_0^*,$$ we get $$B_0^*=\frac{K}{r}\cdot \frac{e^{\int_{0}^1 r-\mu(s)ds}-1}{\int_{0}^1e^{\int_{0}^sr-\mu(u)du}ds }.$$
Thus 
\begin{equation} \label{Phi}
	\Phi(t)=\varphi(t,0,x_0^*)=\frac{K}{r}\cdot\frac{e^{\int_{0}^t r-\mu(s)ds}}{\int_{0}^te^{\int_{0}^sr-\mu(u)du}ds+\frac{\int_{0}^1e^{\int_{0}^sr-\mu(u)du}ds}{e^{\int_{0}^1 r-\mu(s)ds}-1}}
\end{equation}
is the unique, non-constant, and 1-periodic solution when $r>\int_0^1 \mu(s)ds$. By Theorem 4.22 in \cite{JK}, the local stability of $\Phi(t)$ is guaranteed when the Poincare map $\pi(B_0):=\phi(1,B_0)$ satisfies $\pi'(B_0*)<1$, which is a condition equivalent to $r>\int_0^1 \mu(s)ds$. 

Next we show the global stability of $\Phi(t)$. Given another initial value $y_0>0$, we have
\begin{align*}
&\bigl|\frac{1}{\varphi(t,0,x_0^*)}-\frac{1}{\varphi(t,0,y_0)}\bigr|\\
=&e^{\int_{0}^t \mu(u)-rdu}\cdot \frac{r\int_{0}^te^{\int_{0}^sr-\mu(u)du}ds+\frac{K}{x_0^*}}{K}-e^{\int_{0}^t \mu(u)-rdu}\cdot \frac{r\int_{0}^te^{\int_{0}^sr-\mu(u)du}ds+\frac{K}{y_0^*}}{K}\\
=&e^{\int_{0}^t \mu(u)-rdu}\bigl|\frac{1}{x_0^*}-\frac{1}{y_0^*}\bigr|\\
\le& e^{t\int_{0}^1 \mu(u)-rdu}\bigl|\frac{1}{x_0^*}-\frac{1}{y_0^*}\bigr|
\end{align*}

If $\int_{0}^1 r-\mu(u)du>0$, we have that $e^{\int_{0}^1 \mu(u)-rdu}<1$, thus, $\lim\limits_{t\rightarrow\infty} \bigl|\frac{1}{\varphi(t,0,x_0^*)}-\frac{1}{\varphi(t,0,y_0)}\bigr|= 0$.
Notice that 
\begin{align*}
\bigl|\varphi(t,0,x_0^*)-\varphi(t,0,y_0^*)\bigr|&=\bigl|\frac{1}{\varphi(t,0,x_0^*)}-\frac{1}{\varphi(t,0,y_0)}\bigr|\cdot\bigl|\varphi(t,0,y_0)\varphi(t,0,x_0^*)\bigr|
\end{align*}
Thus, with $\lim\limits_{t\rightarrow\infty} |\varphi(t,0,x_0^*)|=\lim\limits_{t\rightarrow\infty}\varphi(t,0,y_0)|= K$, we have that $\lim\limits_{t\rightarrow\infty} |\varphi(t,0,x_0^*)-\varphi(t,0,y_0)|= 0$ and obtain the global stability of $\Phi(t)$ under the condition that $\int_{0}^1 r-\mu(u)du>0$.

If $r< \int_{0}^1 \mu(s)ds$, with $B(0)> 0$, we have $\phi(t,B(0))\ge 0$ and
\begin{flalign*}
    \phi(t,B(0))&\le B(0)e^{\int_0^{\lfloor t \rfloor}r-\mu(s)ds}\cdot e^{\int_0^{t-\lfloor t \rfloor} r-\mu(s)ds}\\
    & \le B(0)\bigl(e^{\int_0^1 r-\mu(s)ds}\bigr)^{{\lfloor t \rfloor}} \rightarrow 0.
\end{flalign*}
Thus the zero solution is globally asymptotically stable.
\end{proof}

\section{More Simulation Results}
\begin{figure}
\begin{subfigure}{0.48\textwidth}
    \includegraphics[width=0.95\linewidth]{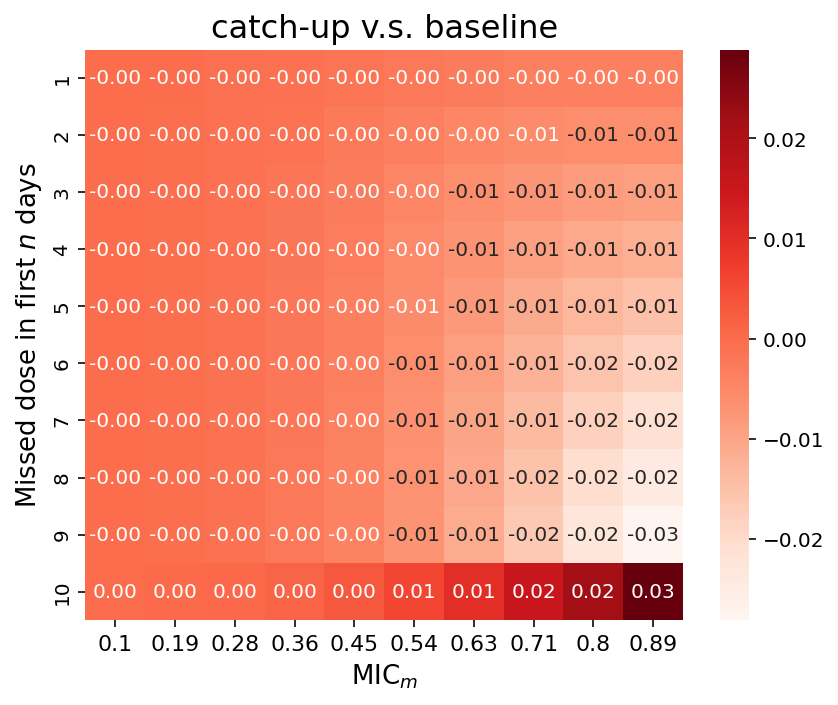}
    \caption{Mutation-induced Resistance}
\end{subfigure}
\begin{subfigure}{0.48\textwidth}
    \includegraphics[width=0.95\linewidth]{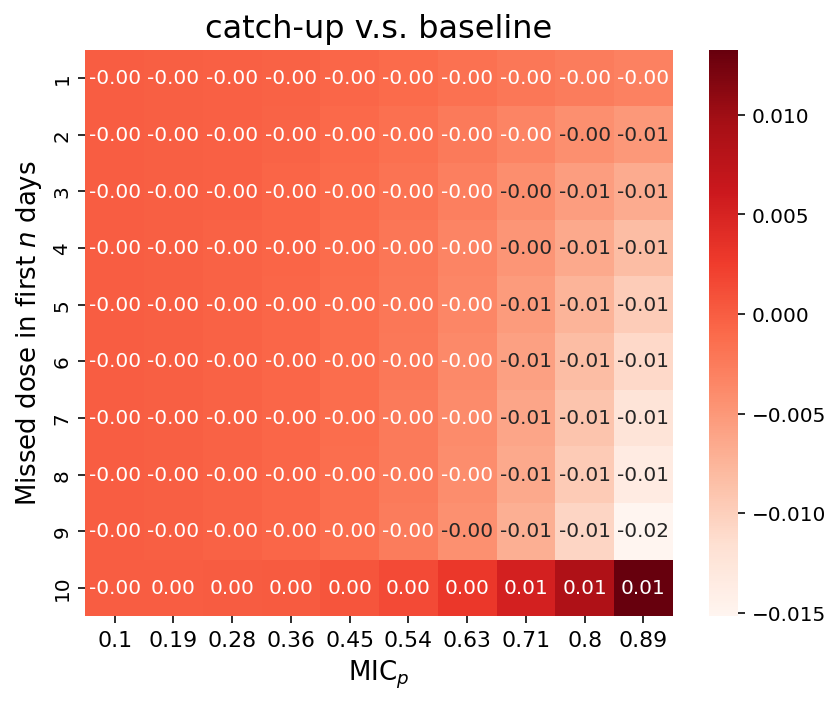}
    \caption{Plasmid-induced Resistance}
\end{subfigure}

\begin{subfigure}{0.48\textwidth}
    \includegraphics[width=0.95\linewidth]{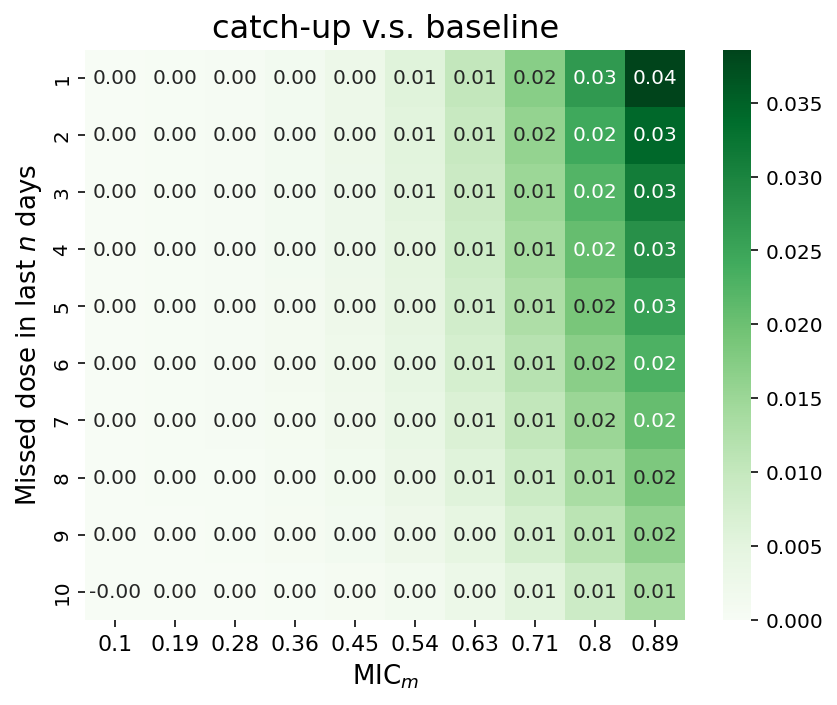}
    \caption{Mutation-induced Resistance}
\end{subfigure}
\begin{subfigure}{0.48\textwidth}
    \includegraphics[width=0.95\linewidth]{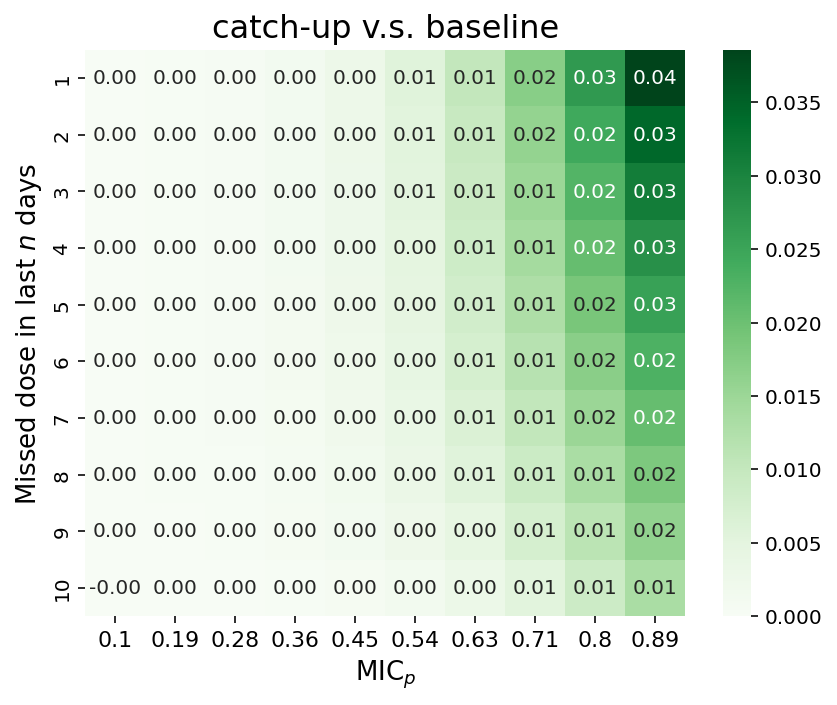}
    \caption{Plasmid-induced Resistance}
\end{subfigure}

\begin{subfigure}{0.48\textwidth}
    \includegraphics[width=0.95\linewidth]{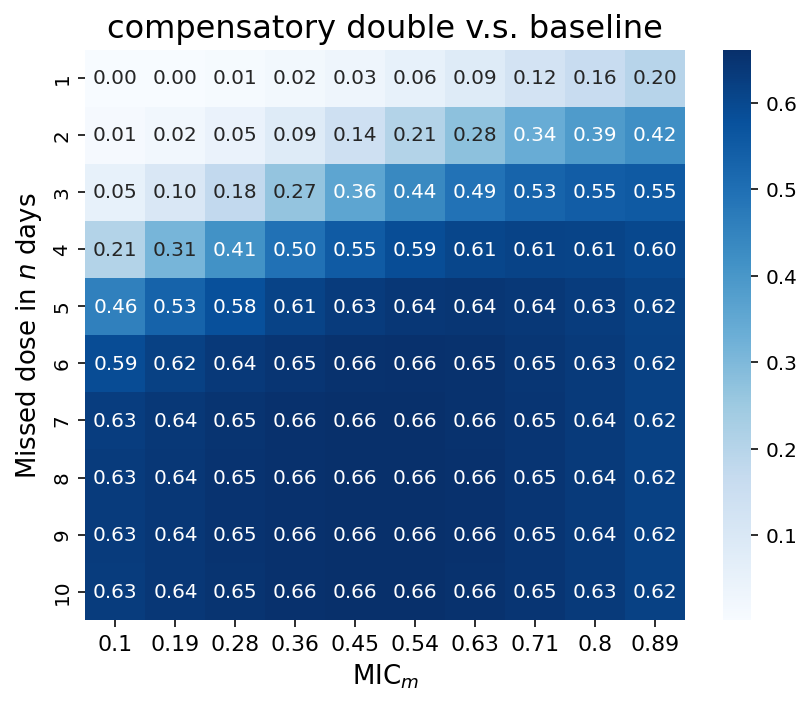}
    \caption{Mutation-induced Resistance}
\end{subfigure}
\begin{subfigure}{0.48\textwidth}
    \includegraphics[width=0.95\linewidth]{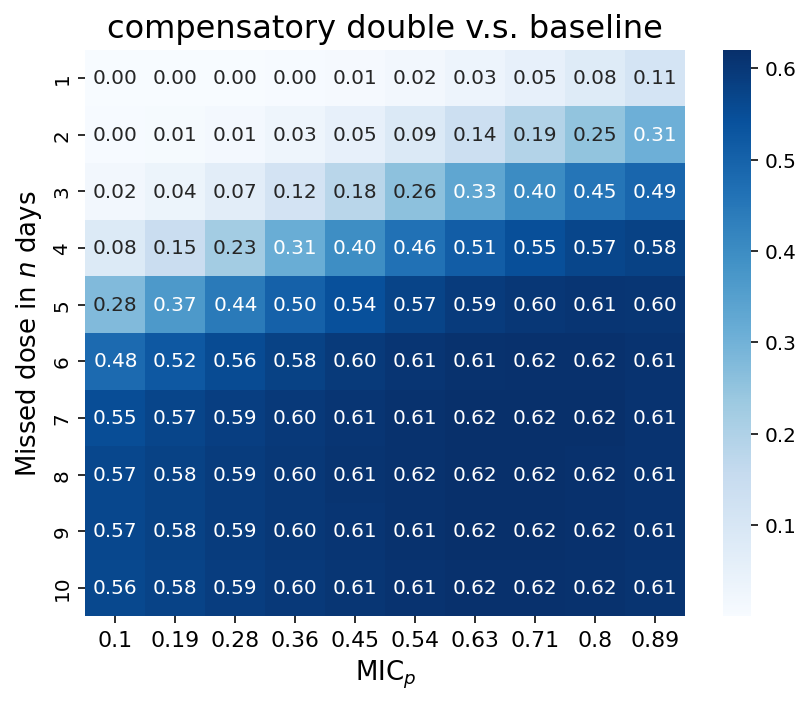}
    \caption{Plasmid-induced Resistance}
\end{subfigure}
    \caption{The Impact of Make-Up Protocols on Final Bacterial Load Following Missed Doses. The figures present the difference in final total bacterial load between various make-up protocols and a twice-daily baseline strategy. Results are shown for both Model \Mtwo{} and \Mthree{}, and are evaluated for missed doses occurring either early or late in the treatment regimen.}
    \label{fig:heatmap}
\end{figure}

\bibliographystyle{siamplain}